\documentclass[final,onecolumn,cm]{llncs}

\newif\iffull
 \fulltrue

\newif\ifshort
\shorttrue

\usepackage[english]{babel}
\usepackage{times}

\usepackage{amsmath}
\usepackage{amsfonts,mathrsfs,amssymb}
\usepackage{amsthm}
\usepackage{mathtools}
\usepackage{stmaryrd}

\usepackage{multirow}

\usepackage{paralist}

\usepackage{booktabs}

\usepackage[final]{graphicx}
\usepackage[final]{listings}
\usepackage{pdfsync}

\newcommand{\keywords}[1]{\par\addvspace\baselineskip
\noindent\keywordname\enspace\ignorespaces#1}

\newtheorem{verificationproblem}{Verification Problem}
\newcounter{custthm}
\newtheorem{customtheorem}[custthm]{Theorem}

\newenvironment{proof*}
    {\begin{proof}[Proof \textup(Sketch\textup)]}
    {\end{proof}}

\usepackage[final,
             bookmarks,
             bookmarksopen,
             colorlinks,
             final,
             linkcolor=blue,
             citecolor=brown,
             backref,
             pdfstartview=FitH ]%
{hyperref}

\usepackage[nompar,nosign]{commenting}

\colorlet{defaultcol}{black!90!yellow}
\declareauthor{lo}{Luke}{red}
\authorcommand{lo}{comment}
\declareauthor{jk}{Jonathan}{blue!90!black}
\authorcommand{jk}{comment}

\input{macros}
\newcommand{\N}{\mathbb{N}}

\newcommand{\ceil}[1]{\ensuremath{\lceil #1\rceil}}
\newcommand{\disjointunion}[0]{\dunion}
\newcommand{\mset}[1]{\ensuremath{\left[#1\right]}}
\newcommand{\implied}[0]{\Leftarrow}
\newcommand{\intersect}[0]{\cap}
\newenvironment{mathprooftree}
  {\varwidth{.9\textwidth}\centering\leavevmode}
  {\DisplayProof\endvarwidth}

\newcommand{\bigslant}[2]{{\raisebox{.15em}{$#1$}\hspace{-1mm}\left/\raisebox{-.15em}{$#2$}\right.}}

\newcommand{\NonT}{\mathcal{N}}
\newcommand{\PPL}{\mathcal{L}}
\newcommand{\Rules}{\ensuremath{\mathcal{R}}}
\newcommand{\eqvI}[0]{\ensuremath{\simeq_I}}
\newcommand{\RHS}[0]{\ensuremath{\text{RHS}}}
\newcommand{\CommWords}{\ensuremath{\bigslant{(\Sigma \union \NonT)^{*}}{\eqvI}}}
\newcommand{\CommTermWords}{\ensuremath{\bigslant{\Sigma^{*}}{\eqvI}}}
\newcommand{\CommNonTermWords}{\ensuremath{\bigslant{\NonT^{*}}{\eqvI}}}
\newcommand{\ComSigma}[0]{{\ensuremath{\Sigma^{\text{com}}}}}
\newcommand{\NComSigma}[0]{{\ensuremath{\Sigma^{\neg\text{com}}}}}
\newcommand{\ComN}[0]{{\ensuremath{\NonT^{\text{com}}}}}
\newcommand{\NComN}[0]{{\ensuremath{\NonT^{\neg\text{com}}}}}

\newcommand\calG{\mathcal{G}}

\newcommand{\toSF}[0]{\ensuremath{\to_{\text{seq}}}} 
\newcommand{\toCF}[0]{\ensuremath{\to_{\text{con}}}} 
\newcommand{\toSM}[0]{\ensuremath{\to_{\text{seq}'}}} 
\newcommand{\toCM}[0]{\ensuremath{\to_{\text{con}'}}} 

\newcommand\snd[2]{{{#1} \mathbin{\text{!}} {#2}}}
\newcommand\rec[2]{{{#1} \mathbin{\text{?}} {#2}}}
\newcommand{\geqShape}{\geq_{\text{shape}}}
\newcommand{\greaterShape}{>_{\text{shape}}}

\newcommand{\cproc}[1]{{#1}}
\newcommand{\cchan}[2]{\mcchan{\mset{#1}}{#2}} 
\newcommand{\mcchan}[2]{{#1}^{#2}} 
\newcommand{\Queue}{\ensuremath{\mathit{Queue}}}
\newcommand{\Queues}{\ensuremath{\mathit{Queues}}}
\newcommand{\Config}{\ensuremath{\mathit{Config}}}
\newcommand{\Control}{\ensuremath{\mathit{Control}}}

\newcommand{\Chan}{\ensuremath{\mathit{Chan}}}
\newcommand{\TermCache}[0]{\ensuremath{\mathit{TermCache}}}
\newcommand{\NonTermCache}[0]{\ensuremath{\mathit{NonTermCache}}}
\newcommand{\MixedCache}[0]{\ensuremath{\mathit{MixedCache}}}
\newcommand{\Cache}[0]{\ensuremath{\mathit{Cache}}}
\newcommand{\CallStack}{\ensuremath{\mathit{CallStack}}}
\newcommand{\ControlState}{\ensuremath{\mathit{ControlState}}}
\newcommand{\DelayedControl}[0]{\ensuremath{\mathit{DelayedControl}}}
\newcommand{\InControl}[0]{\ensuremath{\mathit{NormalControl}}}

\renewcommand{\Msg}{\ensuremath{\mathit{Msg}}}
\newcommand{\ChanPar}{\ensuremath{\mathop{\,\lhd \,}}}

\newcommand{\Reach}[0]{\ensuremath{\mathit{Reach}}}
\newcommand{\Pred}[0]{\ensuremath{\mathit{Pred}}}
\newcommand{\seqpar}[1]{\llparenthesis #1\rrparenthesis}
\newcommand{\seqM}[0]{\overline{\M}}
\newcommand{\simulated}[0]{\preccurlyeq}
\newcommand{\simulates}[0]{\succcurlyeq}
\newcommand{\simulatedS}[0]{\simulated_{\text{S}}}
\newcommand{\simulatedC}[0]{\simulated_{\text{C}}}
\newcommand{\simulatesS}[0]{\simulates_{\text{S}}}

\usepackage[final]{listings}

\usepackage{xcolor}
\colorlet{keyword}{blue!50!black}
\colorlet{atom}{red!50!black}
\colorlet{module}{green!30!black}
\colorlet{comment}{black!70}
\colorlet{coderules}{black!50}
\colorlet{lineno}{black!50}

\makeatletter
\newcommand{\srcsize}{\@setfontsize{\srcsize}{7.625pt}{7.625pt}}
\newcommand{\srcinsize}{\@setfontsize{\srcsize}{8.5pt}{8.5pt}}
\makeatother

\lstdefinelanguage{CoreErlang}
  {morekeywords={fun,and,case,letrec,let,in,catch,div,end,exit,export,halt,%
      if,import,link,make_ref,module,monitor,of,or,receive,self,send,spawn,throw,to,%
      unlink%
      },%
   morekeywords={[2]error,false,nil,ok,true,undefined,pid},%
   otherkeywords={!},%
   morecomment=[l]\%,%
   morestring=[b]"%
  }[keywords,comments,strings]%

\lstdefinestyle{sans}{
    xleftmargin=20pt,
    tabsize=4,
    showstringspaces=false,
    columns=[l]flexible,
    breaklines,
    fontadjust,
    numbers=left,
    numberstyle={\tiny\color{lineno}},
    literate={
        {->}{{$\hspace{-1pt}\rightarrow$}}2
        {...}{{$\ldots$}}3
    },
    basicstyle={\sffamily\srcsize}, 
    keywordstyle={\color{keyword}\bf\srcsize},
    keywordstyle={[2]\color{atom}\srcsize},
    commentstyle={\scriptsize\color{comment}},
    emphstyle={\color{atom}\srcsize},
    emphstyle={[2]\color{module}},
    moredelim=[is][emphstyle]{\#}{\#},
    moredelim=[is][emphstyle2]{@}{@},
    escapechar=§,mathescape
}
\lstdefinestyle{boxed}{frame=single,frameround=tttt,backgroundcolor=\color{keyword!5}}
\lstdefinestyle{head}{style=boxed, frame=tlr, frameround=tfft}
\lstdefinestyle{middle}{style=boxed, frame=lr, firstnumber=last}
\lstdefinestyle{tail}{style=boxed, frame=lrb, frameround=fttf, firstnumber=last}

\lstdefinestyle{inl}{columns=[c]flexible, basicstyle={\sffamily\srcinsize}}

\newcommand{\erl}[2][]{\ifmmode\expandafter\text\fi{\lstinline[style=inl,#1]{#2}}}

\lstnewenvironment{erlang}[1][]{
 \lstset{language=CoreErlang, style=sans, #1}
}{}
\lstnewenvironment{erlang*}[1][]{
	\lstset{language=CoreErlang, style=sans, numbers=none, #1}
}{}
\lstnewenvironment{erlango}[1][]{
  \lstset{language=CoreErlang, style=sans, style=boxed, frame=tlr,%
  frameround=ttff, #1} }{}

\allowdisplaybreaks

\usepackage{tikz}
\usetikzlibrary{shapes,arrows,positioning,fit,matrix,automata}

\usepackage{pdfsync}
\usepackage[normalem]{ulem}
\usepackage{varwidth,bussproofs}

\usepackage[numbers]{natbib}

\begin{document}

\title{Safety Verification of Asynchronous Pushdown Systems with Shaped Stacks}
\titlerunning{Verifying Asynchronous Pushdown Systems with Shaped Stacks}

\author{Jonathan Kochems \qquad \qquad C.-H.~Luke Ong}
\institute{University of Oxford}
\maketitle 

\begin{abstract}
In this paper, we study the program-point reachability problem of concurrent pushdown systems that communicate via unbounded and unordered message buffers. Our goal is to relax the common restriction that messages can only be retrieved by a pushdown process when its stack is empty. 
We use the notion of partially commutative context-free grammars to describe a new class of asynchronously communicating pushdown systems with a mild shape constraint on the stacks for which the program-point cover\-a\-bil\-i\-ty problem remains decidable. 
Stacks that fit the shape constraint may reach arbitrary heights; further a process may execute any communication action (be it process creation, message send or retrieval) whether or not its stack is empty. 
This class extends previous computational models studied in the context of a\-syn\-chro\-nous programs, and enables the safety verification of a large class of message passing programs.
\end{abstract}
\section{Introduction}

The safety verification of concurrent and distributed systems, such as client-server environments, peer-to-peer networks and the myriad web-based applications, is an important topic of research. We consider \emph{asynchronously communicating pushdown systems} (ACPS), a model of computation for such systems suitable for the algorithmic analysis of the reachability problem. Each process of the model is a pushdown system; processes may be spawned dynamically and they communicate asynchronously via a number of unbounded message buffers which may be ordered or unordered. In order to obtain a decision procedure for reachability, some models restrict the retrieval (or, dually, the sending) of messages or the scheduling of tasks, allowing it to take place only when the call stack is empty. 

Can these restrictions on call stacks be relaxed? Unfortunately\footnote{Any analysis that is both context-sensitive and synchronisation-sensitive is undecidable \cite{Ramalingam:2000}.}
some form of constraint on the call stacks in relation to the communication actions is unavoidable. Inspired by the work on asynchronous procedure calls \cite{Sen:2006,Jhala:2007,Ganty:2012}, we consider processes that communicate asynchronously via a fixed number of unbounded and unordered message buffers which we call channels. Because channels are unordered, processes cannot observe the precise sequencing of such concurrency actions as message send and process creation; however, the sequencing of other actions, notably blocking actions such as message retrieval which requires synchronisation, is observable. If the behaviour of a process is given by its action sequences, then we may postulate that certain actions \emph{commute} with each other (over sequential composition) while others do not. To formalise these assumptions, we make use of \emph{partially commutative context-free grammars} (PCCFG) \cite{Czerwinski:09}, introduced recently by \citeauthor{Czerwinski:09} as a study in process algebra. A PCCFG is just a context-free grammar equipped with an irreflexive symmetric relation, called \emph{independence}, over an alphabet $\Sigma$ of terminal symbols, which precisely captures the symbols that \emph{commute} with each other. In our model, a process is described by a PCCFG that generates the set of its action sequences; terminal symbols represent concurrency and communication actions, while the non-terminal symbols represent procedure calls; and there is an induced notion of commutative procedure calls. 
With a view to deciding reachability, a key innovation of our work is to summarise the effects of the commutative procedure calls on the call stack. Rather than keeping track of the contents of the stack, we precompute the actions of those procedure calls that produce only commutative side-effects, and store them in caches on the call stack. The non-commutative procedure calls, which are left on the stack \emph{in situ}, act as separators for the caches of commutative actions. As soon as the top non-commutative non-terminal on the stack is popped, which may be triggered by a concurrency action, the cache just below it is unlocked, and all the cached concurrency actions are then despatched at once. 

In order to obtain a decision procedure for (a form of reachability called) \emph{coverability}, we place a natural constraint on the shape of call stacks: at all times, no more than an \emph{a priori} fixed number of \emph{non-commutative non-terminals} may reside in the stack. Note that because the constraint does not apply to commutative non-terminals, call stacks can grow to arbitrary heights. Thanks to the shape constraint, we can prove that the coverability problem is decidable by an encoding into well-structured transition systems. To our knowledge, this class extends previous computational models studied in the context of asynchronous programs. Though our shape constraint is semantic, we give a simple sufficient condition which is expressed syntactically, thus enabling the safety verification of a large class of message-passing programs. 

\begin{example}\label{eg:rw-pattern}
In Figures~\ref{fig:example:resource_and_distributor} and~\ref{fig:example:worker} we give an example program written in a version of Erlang that employs channels (as opposed to Actor-style mailboxes), implementing a simple replicated workers pattern. It consists of a distributor process that initially spawns a number of workers, sets up a single shared resource, and distributes one task per worker over a one-to-many channel. 
\begin{figure}[t]%
\begin{minipage}{0.49\textwidth}
\begin{erlang}[emph={c1,c2,toResource,toWorkers,toDistributor,isReady,init,false,true,dist,redist,local,which,globalstate,ready,result, unlock_req, lock_req, update, locked, getState,state}]%
main() -> setup_network(), 
          redistribute().

setup_network() ->
  spawn(worker),
  case (*) of
    true  -> setup_network();
    false -> 
      spawn(res_start(init)),
      toResource ! isReady,
      receive toDistributor:
         ready -> ()
      end;
  end, toWorkers ! task.

redistribute() ->
  receive toDistributor:
    redist(Task) -> toWorkers ! Task;
    result(Result) -> print(Result);
  end, redistribute().
\end{erlang}
\end{minipage}
\begin{minipage}{0.49\textwidth}
\begin{erlang}[firstnumber=21,emph={c1,c2,toResource,toWorkers,toDistributor,isReady,init,false,true,dist,redist,local,which,globalstate,ready,result, unlock_req, lock_req, update, locked, getState,state}]
% Resource
res_start(S) = 
  fun() -> toDistributor ! ready, 
           resource(S) 
  end.
resource(S) ->
  receive toResource:
    lock_req -> 
      toWorkers ! locked, 
      resource_locked(S)
  end.

resource_locked(S) ->
  receive toResource:
    unlock_req -> resource(S);
    getState   -> 
      toWorkers ! state(S), 
      resource_locked(S);
    update(X)  -> resource_locked(X)
  end.
\end{erlang}
\end{minipage}
\vspace{-4mm}
\caption{A resource and a task distributor.}
\label{fig:example:resource_and_distributor}
\end{figure}
\begin{figure}[tbh]
\begin{minipage}{0.49\textwidth}
\begin{erlang}[emph={c1,c2,toResource,toWorkers,toDistributor,init,false,true,dist,redist,local,which,globalstate,ready,result,getState,update,locked_req,locked,unlock_req,state}]%
worker() ->
  receive toWorkers:
    Task -> 
      result = do_task(Task),
      toDistributor ! result;
  end, worker().

do_task(Task) ->
  case decompose(Task) of
    local(Task',Int_result) -> 
      Result = do_task(Task'),
      Result' = 
        combine(Result,Int_result)
      return Result';
    redist(Task',Task'') ->  
      Result = do_task(Task'),
      toDistributor ! Redist(Task''),
      return Result;
  end.

combine(res,res')    -> ...
\end{erlang}
\end{minipage}
\begin{minipage}{0.49\textwidth}
\begin{erlang}[firstnumber=22,emph={c1,c2,toResource,toWorkers,toDistributor,init,false,true,dist,redist,local,which,globalstate,ready,result,getState,update,locked_req,locked,unlock_req,state}]%
decompose(Task) -> 
  lock(toResource),
  toResource ! getState,
  ?label("critical"),
  receive toWorkers: 
    state(State) ->
      (Result,Update) = 
        decompose_task(Task, State)
  end,
  toResource ! update(Update),
  unlock(toResource),
  return Result.

lock(C) -> 
  C ! lock_req,
  receive toWorkers:
    locked -> ()
  end.
unlock(C) -> C ! unlock_req.

decompose_task(Task,State) -> ...
\end{erlang}
\end{minipage}
\vspace{-4mm}
\caption{A worker that recursively solves tasks and shares its workload.}
\label{fig:example:worker}
\end{figure}%
Each worker runs a task-processing loop. Upon reception of a task, the worker recursively decomposes it, which involves communicating with the shared resource at each step. Note that the communication of each worker with the resource is protected by a lock. For the worker, the decomposition has two possible outcomes: \begin{inparaenum}[(i)] 
\item the task is partially solved, generating one subtask and an intermediate result or
\label{ex:task_decomp:one}
\item the task is broken down into one subtask and one new distributable task. 
\label{ex:task_decomp:two}
\end{inparaenum}
In case (\ref{ex:task_decomp:one}) the worker recursively solves the subtask and combines the result with the intermediate result. In case (\ref{ex:task_decomp:two}) the worker recursively solves the subtask and subsequently dispatches the newly generated distributable task before returning. When a worker has finished processing a task, it relays the result to the server and awaits a new task to process. We have left the implemention of the functions \erl{decompose_task} and \erl{combine} open; for the purpose of this example we only assume that they do not perform any concurrency actions, but they may be recursive functions.

Note that the call stacks of both the distributor and the workers may reach arbitrary heights, and communication actions may be performed by a process at any stage of the computation, regardless of stack height. For example the worker sends and receives messages at every decomposition, and each recursive call increases the height of the call stack.

An interesting verification question for this example program is whether the locking mechanism for the shared resource guarantees exclusive access to the shared resource for each worker process in its critical section.
\end{example}

\paragraph{A Computational Model.} 
To verify programs such as the above we need a computational model that allows us to model recursive procedure calls, message passing concurrency actions and process creation. Once the obvious abstractions are applied to make the data and message space finite, we arrive at a network of pushdown systems (equivalently context-free grammars) which can communicate asynchronously over a finite number of  channels with unbounded capacity. Since we are interested in a class of such systems with decidable verification problems we assume that channels are unordered (FIFO queues with finite control are already Turing powerful \cite{Brand:83}). 

\paragraph{Outline.}  The rest of the paper is organised as follows. In Section~\ref{sec:acps} we
present our model of asynchronous partially commutative pushdown systems (APCPS), its (standard)
semantics and a verification problem. In Section~\ref{sec:altsem} we investigate an alternative
semantics for APCPS, a corresponding verification problem, and relate it to the verification problem
of Section~\ref{sec:acps}. In Section~\ref{sec:shapedstacks} we introduce the class of APCPS with
shaped stacks and show that the verification problems are decidable for this class. In
Section~\ref{sec:related} we discuss related work and then conclude.%
\iffull%
Owing to space constraints we
have relegated proofs to the appendix.%
\fi

\paragraph{Notation.}
We write $\M[U]$ for the set of multisets over the set $U$, and we use $\mset{\cdot}$ to denote multisets explicitly e.g.~we write $\mset{u,u,v,v}$ to mean the multiset containing two occurrences each of $u$ and $v$. Given multisets $M_1$ and $M_2$, we write $M_1 \oplus M_2$ for the multiset union of $M_1$ and $M_2$. We write $U^*$ for the set of finite sequences over $U$, and let $\alpha, \beta, \gamma, \mu, \nu, \ldots$ range over $U^{*}$. We define the \emph{Parikh image} of $\alpha \in U^*$ to be the multiset over $U$, 
$\M_U(\alpha) : u \mapsto |\{i \mid \alpha(i) = u \}|$; 
we drop the subscript and write $\M(\alpha)$ whenever it is clear from the context. We order multisets in the usual way: $M_1 \leq_{\M} M_2$ just if for all $u$, $M_1(u) \leq M_2(u)$. 
Let $M \in \M[U]$ and $U_0 \subseteq U$. We define $M \restriction U_0$ to be the multiset $M$ restricted to $U_0$ i.e.~$(M \restriction U_0) : u \mapsto M(u)$ if $u \in U_0$, and 0 otherwise. We write $U \uplus V$ for the disjoint union of sets $U$ and $V$.

\section{Asynchronous Communicating Pushdown Systems}
\label{sec:acps}
In this section we introduce our model of concurrency, \emph{asynchronous partially commutative pushdown systems}. Processes are modelled by a variant of context-free grammars, which distinguish commutative and non-commutative concurrency actions. Communication between processes is asynchronous, via a fixed number of unbounded and unordered message buffers, which we call \emph{channels}. 

\subsubsection{Preliminaries.}

An \emph{independence relation} $I$ over a set $U$ is a symmetric irreflexive relation over $U$. 
It induces a congruence relation $\eqvI$ on $U^*$ defined as the least equivalence relation $R$ containing $I$ and satisfying: $(\mu, \mu') \in R \; \implies \; \forall \nu_0, \nu_1 \in U^{\ast} : (\nu_0 \, \mu \, \nu_1, \nu_0 \, \mu' \, \nu_1) \in R.$

Let $I$ be an independence relation over $U$. An element $a \in U$ is \emph{non-commutative} (with respect to $I$) just if $\forall b \in U : (a,b) \notin I$ i.e.~$a$ does not commute with any other element. An element $b$ is \emph{commutative} (with respect to $I$) just if for each $c \in U$, if $c$ is not non-commutative then $(c,b) \in I$; intuitively it means that $b$ commutes with all elements of $U$ except those that are non-commutative. We call an independence relation $I$ \emph{unambiguous} if just every element of $U$ is either commutative or non-commutative. 

\begin{definition}
Let $\Sigma$ be an alphabet of terminal symbols 
and $I \subseteq \Sigma \times \Sigma$ an independence relation over $\Sigma$. A \emph{partially commutative context-free grammar} (PCCFG) is a quintuple \changed[lo]{$\calG = (\Sigma, I, \NonT, \Rules, S)$} where $S \in \NonT$ is a distinguished start symbol, and $\Rules$ is a set of rewrite rules of the following types:\footnote{Identifying rules of type (\ref{item:tail-rec}), which is a special case of type (\ref{item:rec}), allows us to distinguish tail-recursive and non-tail recursive calls, which will be handled differently in the sequel, beginning with Definition~\ref{def:alternative_operational_semantics}. 
} let $A \in \NonT$
\begin{inparaenum}[(i)]
\item \label{item:singleton} $A \rightarrow a$ where $a \in \Sigma \cup \makeset{\epsilon}$,
\item \label{item:tail-rec} $A \rightarrow a \, B$ where $a \in \Sigma$, $B \in \NonT$,
\item \label{item:rec} $A \rightarrow B \, C$ where $B,C \in \NonT$.
\end{inparaenum}
\changed[lo]{We refer to each $\rho \in \Rules$ as a $\calG$-\emph{rule}.}
\end{definition}
The (leftmost) derivation relation ${\toSF}$ is a binary relation over $\CommWords$ defined as $X\, \alpha \toSF \beta \, \alpha$ if $X \rightarrow \beta$ is a $\calG$-rule.  Note the derivation relation is defined over the quotient by $\eqvI$, so the words generated are congruence classes induced by $\eqvI$. As usual we denote the $n$-step relation as $\toSF^n$ and reflexive, transitive closure as $\toSF^*$. 

We further define a $k$-index derivation to be a derivation in which every term contains at most $k$ occurrences of non-terminals.  
Recent work \cite{Esparza:ArXiv:2010,Esparza:2011} has shown that for every commutative context-free grammar $\calG$ there exists $k \geq 1$ such that the entire language of $\calG$ can be generated by derivations of index $k$.

PCCFG was introduced by \citeauthor{Czerwinski:09} as a study in process algebra. They investigated \cite{Czerwinski:09} the decidability of bisimulation for a class of processes described by PCCFG where the commutativity of the sequential composition is constrained by an independence relation on non-terminals. We propose to use words generated by PCCFGs to represent the sequence of concurrency actions of processes.

\subsection{Asynchronous Partially Commutative Pushdown Systems}

Our model of computation, asynchronous partially commutative pushdown systems, are in essence PCCFGs equipped with an independence relation over an alphabet $\Sigma$ of terminal symbols, which represent the concurrency actions and program point labels. First some notation. Let $\Chan$ be a finite set of \emph{channel names} ranged over by $c$, $\Msg$ be a finite \emph{message alphabet} ranged over by $m$, and $\PPL$ be a finite set of \emph{program point labels} ranged over by $l, l', l_1$, etc. Further let $\NonT$ be a finite set of non-terminal symbols. We derive an alphabet $\Sigma$ of terminal symbols 
\begin{equation}
\Sigma \; := \; \PPL \union \{\snd{c}{m}, \rec{c}{m} \mid c \in \Chan, m \in \Msg \} \union \{\nu X \mid X \in \NonT\}.
\label{eq:sigma}
\end{equation}
An action of the form $\snd{c}{m}$ denotes the sending of the message $m$ to channel $c$, $\rec{c}{m}$ denotes the {retrieval} 
of message $m$ from channel $c$, and $\nu X$ denotes the spawning of a new process that begins execution from $X$. We will use $a,a',b,$ etc.~to range over $\Sigma$. Our computational model will emit program point labels in its computation, allowing us to pose questions of reachability. We will now define the computational power of our processes in terms of PCCFGs.
 
The words that are generated by a process \emph{qua} PCCFG represent its action sequences. Because channels are unordered, processes will not be able to observe the precise sequencing of concurrency actions such as message send and process creation; however the sequencing of other actions such as message retrieval is observable. Using the language of partially commutative context-free grammar, we can make this sensitivity to sequencing precise by an independence relation on actions. 

\paragraph{An Independence Relation for the Concurrency Actions.}
Let $\Xi \subseteq \Sigma$, we define the independence relation over $\Sigma$ generated by $\Xi$ as
\[
\mathit{IndRel}_\Sigma(\Xi) := \makeset{(a, a'), (a', a) \mid a, a' \in \Xi, a \neq a'}
\]
Now let $\Sigma^\flat := \PPL \union \makeset{\snd{c}{m} \mid c \in \Chan, m \in \Msg} \union \makeset{\nu X \mid X \in \NonT}$ be the subset of $\Sigma$ consisting of the  program point labels and the send and spawn actions. It is straightforward to see that $\mathit{IndRel}_\Sigma(\Sigma^\flat)$ is, by construction, an unambiguous independence relation over $\Sigma$. Thus $\mathit{IndRel}_\Sigma(\Sigma^\flat)$ allows us to commute all concurrency actions \emph{except} receive. Further we allow program point labels to commute. This is harmless, since our goal is to analyse (a form of) control-state reachability, i.e.~the question whether a particular label can be reached, as opposed to questions that require sequential reasoning such as whether label $l_1$ will be reached before $l_2$ is reached.

We can now lift the independence relation to the non-terminals of a PCCFG $\calG$. Let $I$ be the \changed[lo]{least subset of $(\NonT \cup \Sigma)^2$} such that 
\begin{inparaenum}[(i)]
\item $\mathit{IndRel}_\Sigma(\Sigma^\flat) \subseteq I$, and
\item for all $b \in \Sigma \union \NonT$ and $A \in \NonT$, if $\forall a \in \RHS(A) : (a,b) \in I$ then $\{(A,b),(b,A)\} \subseteq I$,
where $\RHS(A) := \{ {a \in \NonT \cup \Sigma} \mid {A \rightarrow \alpha} \in \calG, a \text{ occurs in } \alpha\}$. 
\end{inparaenum}
We note that $I$, which is well-defined, 
is an unambiguous independence relation over $\NonT \union \Sigma$. Thus we can partition both $\Sigma$ and $\NonT$ into $\ComSigma$ and $\ComN$, the \emph{commutative} actions and non-terminals respectively, and $\NComSigma$ and $\NComN$ their \emph{non-commutative} counterparts respectively.

We can now define our model of computation.
\begin{definition} \label{def:APCPS}
Assume $\PPL, \Chan, \Msg$ and $\NonT$ as introduced earlier, and the derived alphabet $\Sigma$ of terminals as defined in (\ref{eq:sigma}). An \emph{asynchronous partially commutative pushdown system} (APCPS) is just a PCCFG $\calG = (\Sigma, I, \NonT, \Rules, S)$.

\smallskip

\noindent\emph{Henceforth we fix $\PPL, \Chan, \Msg$ and $\NonT$, and the derived (\ref{eq:sigma}) alphabet $\Sigma$ of terminals.} 
\end{definition}

\subsection{Standard Semantics}
The operational semantics is given as a transition system.
A configuration of the system is a pair, consisting of a parallel composition of processes and a set of channels.
We represent the state of a single process as an element of $\Control := \CommWords$. The derivation relation of PCCFGs, $\toSF$, defines how processes make \emph{sequential} transitions. 
Processes interact concurrently by message passing via a fixed set of unbounded and unordered channels.

\begin{definition}[Standard Concurrent Semantics] \label{def:APCPS-semantics}
The \emph{configurations} are elements of $\M[\Control] \times (\Chan \rightarrow \M[\Msg])$. For simplicity, we write a configuration (say) $(\mset{\alpha, \beta}, \{ c_1 \mapsto \mset{m_a,m_b,m_b}, c_2 \mapsto \mset{}\})$ as ${\cproc{\alpha} \parallel \cproc{\beta} \ChanPar \cchan{m_a,m_b,m_b}{c_1}, \cchan{}{c_2}}$. 
We abbreviate a set of processes running in parallel as $\Pi$ and a set of channels by $\Gamma$ with names in $\Chan$. The operational semantics for APCPS, a binary relation $\toCF$ over configurations, is then defined by induction over the rule:
\begin{equation}
\begin{mathprooftree}\label{def:conc_faith_inter}
  \AxiomC{
    $\alpha \toSF \alpha'$
  }
  \UnaryInfC{
        $\cproc{\alpha}   \parallel \Pi \ChanPar \Gamma  \toCF
                \cproc{\alpha'}  \parallel \Pi \ChanPar \Gamma$
    }
\end{mathprooftree}
\end{equation}
and the following axioms: let $m \in \Msg, c \in \Chan, l \in \PPL$ and $X \in \NonT$
\begin{align} 
    \label{def:conc_faith_rec}
    \cproc{(\rec{c}{m})\,\alpha}   \parallel \Pi \ChanPar \mcchan{(\mset{m} \oplus q)}{c}, \Gamma &\toCF
    \cproc{\alpha}                 \parallel \Pi \ChanPar \mcchan{q}{c}, \Gamma\\
    \label{def:conc_faith_send}
    \cproc{(\snd{c}{m})\,\alpha}  \parallel \Pi \ChanPar \mcchan{q}{c}, \Gamma &\toCF
    \cproc{\alpha}                \parallel \Pi \ChanPar \mcchan{(\mset{m} \oplus q)}{c}, \Gamma\\
  \label{def:conc_faith_label}
    \cproc{l\,\alpha}  \parallel \Pi \ChanPar \Gamma &\toCF
    \cproc{\alpha}     \parallel \Pi \ChanPar \Gamma\\
    \label{def:conc_faith_spawn}
    \cproc{(\nu X)\,\alpha}  \parallel \Pi \ChanPar \Gamma &\toCF
    \cproc{\alpha}         \parallel \cproc{X} \parallel \Pi \ChanPar \Gamma.
\end{align}
\end{definition}
\changed[lo]{The \emph{start configuration} is $S \ChanPar \emptyset$.} We define a partial order on configurations: $\Pi \ChanPar \Gamma \, \leq \, \Pi' \ChanPar \Gamma'$ just if \changed[lo]{$\Pi \leq_{\M} \Pi'$} and for every $c \in \Chan$, \changed[lo]{$\Gamma(c) \leq_{\M} \Gamma'(c)$}.

\subsection{Program-Point Coverability}
In the sequential setting of (ordinary) pushdown systems, the control-state reachability problem is of central interest. In our notation, it asks, given a control-state $A$, if it is possible to reach a process-configuration $\cproc{A\, \alpha}$ where $A$ is the control-state and $\alpha$ is some call stack. It should be clear that an equivalent problem is to ask whether $\cproc{l \, \alpha}$ is reachable, where $l$ is a program-point label. We prefer a formulation that uses program-point labels because it simplifies our argument (and is equi-expressive).

In the concurrent setting, we wish to know whether, given an APCPS and program-point labels $l_1,\ldots, l_n$, there exist call stacks $\alpha_1, \ldots, \alpha_n$ and channel contents $\Gamma$ such that the configuration $\cproc{l_1 \, \alpha_1} \parallel \cdots \parallel \cproc{l_n \, \alpha_n} \ChanPar \Gamma$ is $\toCF$-reachable, possibly in parallel with some other processes. Note that this question allows us to express not just control-state reachability queries but also mutual exclusion properties. We state the problem of program-point coverability more formally as follows.

\begin{verificationproblem}[Program-Point Coverability] \label{prob:vp-cover} \em
Given an APCPS $\calG$ and program point labels $l_1, \ldots, l_n$, a tuple $(\calG; l_1, \ldots, l_n)$ is a yes-instance of the \emph{program-point coverability} problem just if there exist a configuration $\Pi \ChanPar \Gamma$ and $\alpha_1, \ldots, \alpha_n \in \CommWords$ such that $\Pi \ChanPar \Gamma$ is $\toCF$-reachable and $\cproc{l_1 \alpha_1} || \cdots ||\,\cproc{l_n \alpha_n} \ChanPar \emptyset \, \leq \, \Pi \ChanPar \Gamma$.
\end{verificationproblem}

The program-point coverability problem allows us to characterise \lq\lq bad-config\-u\-ra\-tions\rq\rq\ $c_{\text{bad}}$ in terms of program-point labels. We regard a configuration $c$ that covers $c_{\text{bad}}$, in the sense that ($c_{\text{bad}} \leq c$), also as \lq\lq bad\rq\rq. Using program-point coverability, we can express whether any such configuration is reachable 

\begin{example}
Consider the program in Figures \ref{fig:example:resource_and_distributor} and  \ref{fig:example:worker} and call it $P$. The problem of whether each worker has exclusive access to the shared resource in its critical section is expressible as a program-point coverability problem. A bad configuration is one in which two worker processes are executing the line marked by \erl{?label("critical")}. We can thus see that $(P;$\erl{?label("critical")}$,$\erl{?label("critical")}$)$ is an instance of the program-point coverability problem; a no answer implies mutual exclusion, a yes answer tells us that two worker processes can be simultaneously within their critical section.
\end{example}

The program-point coverability problem is undecidable for unconstrained APCPS. In fact APCPS is Turing powerful: it is straightforward to simulate a system with two synchronising pushdown systems.

\section{An Alternative Semantics for APCPS}
\label{sec:altsem}

In this section we present an alternative semantics for APCPS which captures enough information to solve the program-point coverability problem. The key idea is to summarise the effects of the commutative non-terminals. 
In the alternative semantics, rather than keeping track of the contents of the call stack, we precompute the actions of those procedure calls that produce only \emph{commutative} side-effects, i.e.~sends, spawns and program point labels, and store them in caches on the call stack. The non-com\-mu\-ta\-tive procedure calls, which are left on the call stack, then act as separators for the caches of commutative side-effects. As soon as the top non-commutative non-terminal on the stack is popped, which may be triggered by a concurrency action, the cache just below it is unlocked. The cached actions are made effective instantaneously. This is enough to ensure a precise correspondence between the program-point coverability problem for APCPS and a corresponding coverability problem for our alternative semantics.

\subsubsection{An Alternative Semantics.}
First we introduce a representation of the states of a process.
Let $k \in \N \union \{\infty\}$.
\begin{align*}
                     \TermCache          :=\,& \M[\ComSigma] 
 \hspace{10mm}       \MixedCache         :=\, \M[\ComSigma \union \ComN] \\
                     \NonTermCache       :=\,& \M[\ComN] 
 \hspace{10mm}       \Cache              :=\, \TermCache \disjointunion \MixedCache\\
                     \CallStack^{\leq k} :=\,& (\NComN \cdot \Cache)^{\leq k} \\
                     \DelayedControl     :=\,& \TermCache
                                                  \disjointunion \MixedCache
                                                  \disjointunion \NonTermCache\\
                     \InControl          :=\,& (\NonT\cdot \Cache) \disjointunion 
                                                  (\Sigma \cdot \NonT \cdot \Cache) \disjointunion 
                                                  (\Sigma \cdot \Cache)\\
                    \ControlState        :=\,& \InControl  \disjointunion \DelayedControl\\
\gamma,\delta \in    \Control^{\leq k}   :=\,& \ControlState \cdot \CallStack^{\leq k}\\
                    \Queue               :=\,& \M[\Msg]
 \hspace{10mm}      \Queues              :=\, \Chan \to \Queue\\
                    \Config^{\leq k}     :=\,& \M\left[\Control^{\leq k}\right] \times \Queues
\end{align*}

Note that we assume the equality $\epsilon = \emptyset$ to simplify notation. 
We write $\Control^\M := \Control^{\leq \infty}$ and $\CallStack^\M := \CallStack^{\leq \infty}$. \changed[jk]{\footnote{Defining $\Cache$ as a distjoint union enables a definition by cases according to the type of cache, thus rendering $\toCM$ monotone with respect to an ordering.}} 

\begin{definition}[Alternative Sequential Semantics]\label{def:alternative_operational_semantics}
Let $\calG$ be a PCCFG. We define a transition relation ${\toSM}$ on $\Control^{\M}$ by induction over the following rules:
\begin{align} 
\shortintertext{If $A \rightarrow B\,C$ is a $\calG$-rule, $C$ commutative and $C \toSF^* w \in (\ComN \union \ComSigma)^*$ then}
A\, M\, \gamma &\toSM B\, (\M(w)\, \oplus\, M)\,  \gamma  \label{def:seq_mult_sem_nonblock} 
\shortintertext{If  $A \rightarrow B\,C$ is a $\calG$-rule and $C$ non-commutative then}
A\, M\, \gamma &\toSM B\, C\, M\,  \gamma \label{def:seq_mult_sem_block}
\shortintertext{If  $A \rightarrow a\,B$ is a $\calG$-rule and $a \in \Sigma$
 and $B \in \NonT$ then}
A\, M\, \gamma &\toSM a\, B\, M\,  \gamma \label{def:seq_mult_sem_tailrec}
\shortintertext{If  $A \rightarrow a$ is a $\calG$-rule where $a \in \Sigma \union \{\epsilon\}$ then}
A\, M\, \gamma &\toSM a\, M\,  \gamma \label{def:seq_mult_sem_action}
\end{align}
where $\gamma \in \CallStack^\M$, $M \in \Cache$, and $A,B$ and $C$ range over non-terminals.
\end{definition}

From the alternative sequential semantics, we derive a corresponding alternative concurrent semantics,
using the following notation: for $M \in \M[\ComSigma]$ and $w \in (\ComSigma)^*$
\[
\begin{array}{ll}
\multicolumn{2}{l}{\Gamma \oplus \Gamma' := \{c \mapsto \Gamma(c) \oplus \Gamma'(c) \mid c \in \Chan\}}\\
\Gamma(M) := \left\{c \mapsto \sum_{\snd{c}{m} \in M}M(\snd{c}{m}) \mid c \in \Chan \right\}
\quad &
\Gamma(w) := \Gamma(\M(w))
\\
\Pi(M) := \{\cproc{X} \mapsto M(\nu X) \mid X \in \NonT\}
&
\Pi(w) := \Pi(\M(w))
\\
\end{array}\]

\begin{definition}[Alternative Concurrent Semantics]
We define a binary relation $\toCM$ over $\M[\Control^\M] \times (\Chan \rightarrow \M[\Msg])$ by induction over the following rules:
\begin{align}  
    \shortintertext{If $\gamma \in \InControl\cdot\CallStack^\M$, $\gamma \toSM \gamma'$ then}
      \cproc{\gamma}   \parallel \Pi \ChanPar \Gamma &\toCM 
      \cproc{\gamma'}  \parallel \Pi \ChanPar \Gamma \label{def:con_mult_sem_interleave} 
    \shortintertext{If $(\rec{c}{m})\, \gamma \in \InControl\cdot\CallStack^\M$, $m \in \Msg$ then}
    \cproc{(\rec{c}{m})\, \gamma}   \parallel \Pi \ChanPar  \mcchan{(\mset{m} \oplus q)}{c}, \Gamma&\toCM 
    \cproc{\gamma}                  \parallel \Pi \ChanPar  \mcchan{q}{c}, \Gamma \label{def:con_mult_sem_rec}
    \shortintertext{If $X \in \NonT$, $(\nu X)\, \gamma \in \InControl\cdot \CallStack^\M$ then}
    \cproc{(\nu X)\, \gamma} \parallel \Pi               \ChanPar \Gamma &\toCM 
    \cproc{\gamma}           \parallel \cproc{X} \parallel \Pi  \ChanPar \Gamma \label{def:con_mult_sem_spawn}   
    \shortintertext{If $(\snd{c}{m})\, \gamma \in \InControl\cdot\CallStack^\M$, $m \in \Msg$ then}
    \cproc{(\snd{c}{m})\, \gamma}   \parallel \Pi \ChanPar  \mcchan{q}{c}, \Gamma&\toCM 
    \cproc{\gamma}                  \parallel \Pi \ChanPar  \mcchan{(\mset{m} \oplus q)}{c}, \Gamma \label{def:con_mult_sem_send} 
    \shortintertext{If $l\, \gamma \in \InControl\cdot\CallStack^\M$, $l \in \PPL$ then}
    \cproc{l\, \gamma} \parallel \Pi \ChanPar  \Gamma&\toCM 
    \cproc{\gamma}     \parallel \Pi \ChanPar  \Gamma \label{def:con_mult_sem_label} 
    \shortintertext{If $M\, X \,\gamma \in \DelayedControl \cdot \CallStack^{\M}$,
                  $M \in \TermCache$, 
                  $\Gamma' = \Gamma \oplus \Gamma(M)$, 
                  $\Pi' = \Pi \oplus \Pi(M)$
                  then
              }
    \cproc{M\, X \,\gamma} \parallel \Pi  \ChanPar \Gamma &\toCM 
    \cproc{X\, \gamma}  \parallel \Pi' \ChanPar \Gamma' \label{def:con_mult_sem_disp}\\
    \shortintertext{If $M\, \gamma \in \DelayedControl \cdot \CallStack^{\M}$,
                  $M \in \MixedCache$, 
                  $\Gamma' = \Gamma \oplus \Gamma(M)$, 
                  $\Pi' = \Pi \oplus \Pi(M)$ and 
                  \changed[lo]{$M' = M \restriction {(\ComN \union \PPL)}$}
                  then
              }
    \cproc{M\,  \gamma}  \parallel \Pi  \ChanPar \Gamma &\toCM 
    \cproc{M'\, \gamma}  \parallel \Pi' \ChanPar \Gamma' \label{def:con_mult_sem_non-term}
\end{align}   
\end{definition}
 
The alternative semantics precomputes the actions of commutative non-terminals on the call stacks. This is achieved by rule (\ref{def:seq_mult_sem_nonblock}) in the alternative sequential semantics. The rules (\ref{def:con_mult_sem_disp}) and (\ref{def:con_mult_sem_non-term}) are the concurrent counterparts; they ensure that the precomputed actions are rendered effective at the appropriate moment. Rule (\ref{def:con_mult_sem_disp}) is applicable when the precomputed cache $M$ contains exclusively commutative actions; such a cache denotes a sequence of commutative non-terminals whose computation terminates and generates concurrency actions. Rule (\ref{def:con_mult_sem_non-term}), on the other hand, handles the case where the cache $M$ contains non-terminals. An interpretation of such a cache is a partial computation of a sequence of commutative non-terminals. In this case rule (\ref{def:con_mult_sem_non-term}) dispatches all commutative actions and then blocks. It is necessary to consider this case since not all non-terminals have terminating computations. Thus rule (\ref{def:seq_mult_sem_nonblock}) may non-deterministically decide to abandon the pre-compution of actions. 

We give a variant of the program-point coverability problem tailored to the alternative semantics and show its equivalence with the program-point coverabilily problem.

\begin{verificationproblem}[Alternative Program-Point Coverability] \em
Given an APCPS $\calG$ and a set of program point labels $l_1, \ldots, l_n$, a tuple $(P; l_1, ..., l_n)$ is a yes-instance of the \emph{alternative program-point coverability} problem just if there exist a $\toCM$-reachable configuration $\Pi \ChanPar \Gamma$ such that for every $i \in \makeset{1,\ldots, n}$ there exists $\cproc{\lambda_i \, \gamma_i} \in \Pi$ such that either $\lambda_i = l_i$, or $\lambda_i = M_i$ and $l_i \in M_i$?
\end{verificationproblem}

\iffull
In the appendix we show that the standard semantics \emph{weakly simulates} the alternative semantics for APCPS (Proposition \ref{apx:prop:concurrent_simulation}). Thus for every configuration reachable in the alternative semantics there is a corresponding configuration reachable in the standard semantics. Owing to the nature of precomputations and caches, it is more difficult to relate runs of the standard semantics to those of the alternative semantics. However, in the appendix, we show that for every run in the standard semantics reaching a configuration, there exists a run in the alternative semantics reaching a corresponding configuration (Proposition \ref{apx:prop:conc_reduction_simulation}).
\fi

\begin{theorem}[Reduction of Program-Point Coverability]\label{thm:reduction_coverability}
A tuple  $(P;l_1, \ldots, l_n)$ is a yes-instance of the program-point coverabililty problem if, and only if, $(P;l_1, \ldots, l_n)$ is a yes-instance of the alternative program-point coverability problem.
\end{theorem}
\section{APCPS with Shaped Stacks}
\label{sec:shapedstacks}
In this section we present a natural restriction on the shape of the call stacks of APCPS processes. This shape restriction says that, at all times, at most an \emph{a priori} fixed number of non-commutative non-terminals may reside in the call stack. Because the restriction does not apply to commutative non-terminals, call stacks can grow to arbitrary heights. We show that the alternative semantics for such shape-constrained APCPS gives rise to a well-structured transition system, thus allowing us to show the decidability of the alternative program-point coverability problem. 
\begin{definition}
Define $\Reach_{\toCM} := \{\Pi \ChanPar \Gamma \mid [\cproc{S}] \ChanPar \emptyset \toCM^* \Pi \ChanPar \Gamma\}.$ Let $k \in \N$, we say an APCPS $\calG$ \emph{has $k$-shaped stacks} just if $\Reach_{\toCM} \subseteq \Config^{\leq k}$. An APCPS $\calG$ \emph{has shaped stacks} just if $\calG$ has $k$-shaped stacks for some $k \in \N$.
\end{definition}
It follows from the definition that,  in the alternative semantics, processes of an APCPS with $k$-shaped stacks have the form:
$\gamma \, X_1 \, M_1 \, X_2 \, M_2 \cdots X_j \, M_j$ where $\gamma \in \ControlState$, $X_i \in \NComN$ and $j \leq k$. Relating this
to the standard semantics, processes of an APCPS with $k$-shaped stacks are always of the form
$\alpha \, X_1 \, \beta_1 \, X_2 \, \beta_2  \cdots X_j \, \beta_j$ where $\alpha \in (\NonT \union (\Sigma\cdot\NonT) \union (\Sigma \union \{\epsilon\})) \cdot \ComN^*$ and $\beta_i \in \ComN^*$. It is this shape that lends itself to the name APCPS. Even though the shaped stacks constraint is semantic, we can give a \emph{syntactic} sufficient condition: (the simple proof is omitted.)

\begin{proposition}\label{prop:sufficient_cond}
Let $\calG$ be an APCPS. If there is a well-founded partial order $\geqShape$ such that 
for every $A \in \NonT$ and $B \in \RHS(A) \intersect \NonT$:
\begin{inparaenum}[(i)] 
\item $A \geqShape B$, and 
\item $\exists C \in \NComN : \hbox{$A \rightarrow B\, C$ is a $\calG$-rule} \; \implies \; A \greaterShape B$,
then $\calG$ has shaped stacks. 
\end{inparaenum}
\end{proposition}

\begin{example}
Proposition \ref{prop:sufficient_cond} tells us that the program in Figures \ref{fig:example:resource_and_distributor} and \ref{fig:example:worker} can be modelled by an APCPS with shaped stacks. Non-tail recursive calls are potentially problematic. In our example the recursive call to \erl{setup_network()} in the definition of \erl{setup_network} is non-tail recursive, but only places a send action on the call stack, thus causing no harm.
\ifshort%
The only other non-tail recursive calls occur in \erl{do_task}: the call to \erl{decompose_task} poses no threat since \erl{decompose_task} does not invoke \erl{do_task} again.
The two recursive calls to \erl{do_task} either place procedure calls with send or no concurrent actions on the stack.
\else%
The only other non-tail recursive calls occur in \erl{do_task}: the call to \erl{decompose_task} poses no threat since \erl{decompose_task} does not invoke \erl{do_task} again.
The two recursive calls to \erl{do_task} which either place procedure calls without concurrent actions or send actions on the stack.
\fi%
\end{example}

\subsection{APCPS with Shaped Stacks and Well-Structured Transition Systems}
We will now show the decidability of the alternative program-point coverability problem for APCPS with shaped stacks. First we recall the definition of well-structured transition systems \cite{Finkel:01}.
Let $\leq$ be an ordering over a set $U$; we say $\leq$ is a \emph{well-quasi-order} (wqo) just if for all infinite sequences $u_1,u_2,\ldots$ there exists $i,j$ such that $u_i \leq u_j$. A \emph{well-structured transition system} (WSTS) is a quadruple $(S,\rightarrow,\leq,s_0)$ such that $s_0 \in S$, $\leq$ is a wqo over $S$ and ${\rightarrow} \subseteq S \times S$ is monotone with respect to $\leq$, i.e. if $s \rightarrow s'$ and $s \leq t$ then there exists $t'$ such that $t \rightarrow t'$.

WSTS are an expressive class of infinite state systems that enjoy good model checking properties. A decision problem for WSTS of particular interest to verification is the \emph{coverability problem} i.e.~given a state $s$ is it the case that $s_0 \rightarrow^* s'$ and $s \leq s'$. For $U \subseteq S$ define the sets $\Pred(U) := \{ s \mid s \rightarrow u, u \in U\}$ and $\uparrow U := \{ u' \mid u \leq u', u \in U\}$.
For WSTS the coverability problem is decidable \cite{Finkel:01} provided that for any given $s \in S$ the set $\uparrow\Pred(\uparrow \{s\})$ is effectively computable. Wqos can be composed in various ways which makes decision results for WSTS applicable to a wide variety of infinite state models. In the following we recall a few results on the composition of wqos.%
{\small%
\begin{enumerate}[(WQO-a)]
  \item If $(A_i,\leq_i)$ are wqo sets for $i = 1,...,k$ then $(A_1 \times \cdots \times A_k, \leq_1 \times \cdots \times \leq_k)$ is a wqo set. (\emph{Dickson's Lemma})
  \item If $A$ is a finite set then $(A,=)$ is a wqo set.
  \item If $(A,\leq)$ is a wqo then $(\M[A],\leq_{\M[A]})$ is a wqo set where $M_1 \leq_{\M[A]} M_2$ just if for all $a \in A$ there exists an $a' \geq a$ such that $M_1(a) \leq M_2(a')$ \cite{Wehrman:2006}.  \label{prop:wqo:varhigman}
  \item If $(A, \leq_A)$ and $(B, \leq_B)$ are wqo sets, then $(A \cdot B, \leq_A \cdot \leq_B)$ is a wqo set, where $\gamma \cdot \gamma' \leq_A \cdot \leq_B \delta \cdot \delta$ just if $\gamma \leq_A \delta$ and $\gamma' \leq_B \delta'$. 
  \item If $(A, \leq_A)$ and $(B, \leq_B)$ are wqo set, then $(A \disjointunion B, \leq_A \disjointunion \leq_B)$ is a wqo set, where $a \leq_A \disjointunion \leq_B b$ just if $a,b \in A$ and $a \leq_A b$ or $a,b \in B$ and $a\leq_B b$. 
\end{enumerate}
}
\newcommand{\eqX}[1]{\mathop{=_{#1}}}

\subsection{A Well-Quasi-Order for the Alternative Semantics}
Fix a $k$. Our goal is to construct a well-quasi-order for $\Config^{\leq k}$ as a first step to showing the alternative semantics gives rise to a WSTS for APCPS with shaped stacks.

We order the multi-sets \TermCache, \NonTermCache, \MixedCache\ and $\Queue$ with the multi-set inclusion $\leq_{\M}$ which is a well-quasi-order.
Since \Chan\ is a finite set and $\Queues = \Chan \rightarrow \M[\Msg] \cong \M[\Msg]^{|\Chan|}$ we obtain a well-quasi-order for $\Chan \rightarrow \M[\Msg]$ using a generalisation of Dickson's lemma.
We then compose the wqo of \TermCache\ and \MixedCache\ to obtain a wqo $\leq_{\Cache} := \leq_{\TermCache} \disjointunion \leq_{\MixedCache}$ for \Cache.
For each $j \in \{1\ldots k\}$ we define 
$$X_1 \, M_1\, X_2\, M_2\cdots X_j\, M_j \leq X_1\, M'_1 \,X_2\, M'_2\cdots X_j \, M'_j \quad \hbox{iff} \quad  \forall i : M_i \leq_{\Cache} M'_i$$
which gives a well-quasi-order for $\CallStack^{\leq k}$. 
We obtain a wqo for \DelayedControl\ by composing the wqos of \TermCache, \NonTermCache\ and \MixedCache: 
$$\leq_{\DelayedControl}\, \is\, \leq_{\TermCache} \disjointunion 
 \leq_{\NonTermCache} \disjointunion \leq_{\MixedCache}.$$ 
Since $\Sigma$ and $\NonT$ are finite sets, $(\Sigma,\eqX{\Sigma})$ and $(\NonT,\eqX{\NonT})$ are wqo sets, and so, we can compose a wqo for \InControl:
$$\leq_{\InControl}\, \is \left(\eqX{\Sigma} \cdot \leq_{\Cache}\right) \mathbin{\disjointunion} 
                        \left(\eqX{\Sigma} \cdot \eqX{\NonT} \cdot \leq_{\Cache}\right) \mathbin{\disjointunion}
                        \left(\eqX{\NonT} \cdot \leq_{\Cache}\right).$$
Similarly we can construct wqos for \ControlState\ and $\Control^{\leq k}$ by composition:
\begin{align*}
\leq_{\ControlState}     &:=\,\leq_{\InControl} \disjointunion \leq_{\DelayedControl}\\
\leq_{\Control^{\leq k}} &:=\,\leq_{\ControlState} \cdot \leq_{\CallStack^{\leq k}}.
\end{align*}
As a last step we use (WQO-\ref{prop:wqo:varhigman}) to construct a wqo for $\M\left[\Control^{\leq k}\right]$ which then allows us to define a wqo for $\Config^{\leq k}$ by 
$\leq_{\Config^{\leq k}} := \leq_{\M\left[\Control^{\leq k}\right]} \times \leq_{\Queues}.$

To prove the decidability of the coverability problem for APCPS with shaped stacks, it remains to show that $\toCM$ is monotonic and $\uparrow\Pred(\uparrow \{\gamma\})$ is computable.

\begin{lemma}[Monotonicity]\label{lemma:monotonicity_of_inst_chan}
The transition relation $\toCM$ is monotone with respect to the well-order ${\leq_{\Config^{\leq k}}}$.
\end{lemma}

\begin{corollary}
The transition system ${\left(\Config^{\leq k},\toCM,\leq_{\Config^{\leq k}}\right)}$ 
is a well-structured transition system.
\end{corollary}

To see that $\uparrow\Pred(\uparrow \{\gamma\})$ is computable is mostly trivial;
only predecessors\linebreak generated by rule (\ref{def:seq_mult_sem_nonblock}) are not immediately obvious. Given ${M' \in \Cache}$ we observe that it is enough to be able to compute the set $P_{M'} := {\uparrow\{(C,M) \mid C \in \ComN,}\linebreak{C \toSF^* w, M'' = M \oplus \M(w),  M' \leq_{\M} M''\}}$. Now $C \toSF^* w$ is a computation of a commutative context-free grammar (CCFGs) for which an encoding into Petri nets has been shown by \citeauthor{Ganty:2012} \cite{Ganty:2012}.
Their encoding builds on work by \citeauthor{Esparza:1997} \cite{Esparza:1997} modelling CCFG in Petri nets. Their translation leverages a recent result \cite{Esparza:ArXiv:2010}: every word of a CCFG has a \emph{bounded-index} derivation i.e.~every term of the derivation uses no more than an \emph{a priori} fixed number of occurrences of non-terminals. A budget counter constrains the Petri net encoding of a CCFG to respect boundedness of index; termination of a CCFG computation can be detected by a transition that is only enabled when the full budget is available. This result allows us to compute the set $P_{M'}$ using a backwards coverability algorithm for Petri nets.

\begin{theorem}
  The alternative program-point coverability problem, and hence the program-point coverability problem, for APCPS with $k$-shaped stacks are decidable for every $k \geq 0$.
\end{theorem}
\section{Related Work and Discussion}
\label{sec:related}
\paragraph{Partially Commutative Context-Free Grammars (PCCFG).}
\citeauthor{Czerwinski:09} introduced PCCFG as a study in process algebra \cite{Czerwinski:09}.
They proved that bisimulation is NP-complete for a class of processes extending BPA and BPP \cite{Esparza:1997} where the sequential composition of certain processes is commutative. Bisimulation is defined on the traces of such processes, although there is no synchronisation between processes. In \cite{Czerwinski:Concur:2012} the problem of word reachability for partially commutative context-free languages was shown to be NP-complete. 

\paragraph{Asynchronous Procedure Calls.}
Petri net models for finite state machines that communicate asynchronously via unordered message buffers were first investigated by \citeauthor{Mukund:Icalp:1998} \cite{Mukund:Asian:1998, Mukund:Icalp:1998}. In an influential paper \cite{Sen:2006} in 2006, \citeauthor{Sen:2006} showed that safety verification is decidable for first-order programs with atomic asynchronous methods. Building on this, \citeauthor{Jhala:2007} \cite{Jhala:2007} constructed a VAS that models such asynchronous programs on-the-fly. Liveness properties, such as fair termination and starvation, of asynchronous programs were extensively studied by \citeauthor{Ganty:2009} in \cite{Ganty:2009,Ganty:2012}. 
In our more general APCPS framework, we may view the asynchronous programs considered by \citeauthor{Ganty:2012} in \cite{Ganty:2012} as APCPS running a single \lq\lq scheduler\rq\rq\  process. Task bags can be modelled as channels in our setting and the posting of a task can be modelled by sending a message; the scheduling of a procedure call can be simulated as a receive of a non-deterministically selected channel which unlocks a commutative procedure call defined by rules of types (\ref{item:singleton}) \changed[lo]{and (\ref{item:tail-rec})} and rules of type (\ref{item:rec}) where $C \in \ComN$, in the sense of Definition~\ref{def:APCPS}. It is thus easy to see that APCPS with shaped stacks subsume programs with asynchronous procedure calls. In light of the fact that their safety verification is {\scshape ExpSpace}-complete we can infer that the program-point coverability problem for APCPS with shaped stacks is {\scshape ExpSpace}-hard.

Various extensions of \citeauthor{Sen:2006}'s model \cite{Chadha:2007} and applications to real-world asynchronous task scheduling systems \cite{Geeraerts:2012} have been investigated. From the standpoint of message-passing concurrency, a key restriction of many of the models considered is that messages may only be retrieved by a communicating pushdown process when its stack is empty. The aim of this paper is to relax this restriction while retaining decidability of safety verification.

\paragraph{Communicating Pushdown Systems.}
The literature on communicating pushdown systems is vast. Numerous classes with decidable verification problems have been discovered. Heu{\ss}ner et~al.~\cite{Heussner:2010} studied a restriction on pushdown processes that communicate asynchronously via FIFO channels: a process may send a message only when its stack is empty, while message retrieval is unconstrained. Several other communicating pushdown systems have been explored: parallel flow graph systems \cite{Esparza:2000}, visibly pushdown automata that communicate over FIFO-queues \cite{Babic:2011}, pushdown systems communicating over locks \cite{Kahlon:2009}, and recursive programs with hierarchical communication \cite{Bouajjani:2005,Bouajjani:Popl:2012}.

Verification techniques that over-approximate correctness properties of concurrent pushdown systems have been studied \cite{Flanagan:2003,Henzinger:2003}. Under-approximation techniques typically impose constraints, such as bounding the number of context switches \cite{Torre:2009,Lal:2009}, bounding the number of times a process can switch from a send-mode to receive-mode \cite{Bouajjani:Tacas:2012}, or allowing symbols pushed onto the stack to be popped only within a bounded number of context switches \cite{Torre:2011}. Another line of work focuses on pushdown systems that communicate synchronously over channels, restricting model checking to synchronisation traces that fall within a restricted regular language \cite{Esparza:2011}; this approach has been developed into an effective CEGAR method \cite{Long:2012}.
\subsubsection*{Future Directions and Conclusion.}
We have introduced a new class of asynchronously communicating pushdown systems, APCPS, and shown that the program-point coverability problem is decidable and \textsc{ExpSpace}-hard for the subclass of APCPS with shaped stacks. We plan to investigate the precise complexity of the program-point coverability problem, construct an implementation and integrate it into {\scshape Soter} \cite{Soter:2012,Soter:2013}, a safety verifier for Erlang programs, to study APCPS empirically.

\subsubsection*{Acknowledgments.}
Financial support by EPSRC (research grant EP/F036361/1 and
OUCL DTG Account doctoral studentship for the first author) is gratefully acknowledged. \shortfalse%
\ifshort%
We would like to thank the anonymous reviewers for their detailed comments.  
\else%
We would like to thank Matthew Hague, Subodh Sharma, Michael Tautschnig and Emanuele D'Osualdo for helpful discussions and insightful comments, and the anonymous reviewers for their detailed reports.
\fi%
\bibliographystyle{abbrvnat}
{\small
\bibliography{bibliography/pccfg,%
              bibliography/ref,%
              bibliography/vass,%
              bibliography/asyncpc,%
              bibliography/concpds%
}

\begin{thebibliography}{33}
\providecommand{\natexlab}[1]{#1}
\providecommand{\url}[1]{\texttt{#1}}
\expandafter\ifx\csname urlstyle\endcsname\relax
  \providecommand{\doi}[1]{doi: #1}\else
  \providecommand{\doi}{doi: \begingroup \urlstyle{rm}\Url}\fi

\bibitem[Babic and Rakamaric(2011)]{Babic:2011}
D.~Babic and Z.~Rakamaric.
\newblock Asynchronously communicating visibly pushdown systems.
\newblock Technical Report UCB/EECS-2011-108, UC Berkeley, 2011.

\bibitem[Bouajjani and Emmi(2012{\natexlab{a}})]{Bouajjani:Popl:2012}
A.~Bouajjani and M.~Emmi.
\newblock Analysis of recursively parallel programs.
\newblock In \emph{POPL}, pages 203--214, 2012{\natexlab{a}}.

\bibitem[Bouajjani and Emmi(2012{\natexlab{b}})]{Bouajjani:Tacas:2012}
A.~Bouajjani and M.~Emmi.
\newblock Bounded phase analysis of message-passing programs.
\newblock In \emph{TACAS}, pages 451--465, 2012{\natexlab{b}}.

\bibitem[Bouajjani et~al.(2005)Bouajjani, M{\"u}ller-Olm, and
  Touili]{Bouajjani:2005}
A.~Bouajjani, M.~M{\"u}ller-Olm, and T.~Touili.
\newblock Regular symbolic analysis of dynamic networks of pushdown systems.
\newblock In \emph{CONCUR}, pages 473--487, 2005.

\bibitem[Brand and Zafiropulo(1983)]{Brand:83}
D.~Brand and P.~Zafiropulo.
\newblock On communicating finite-state machines.
\newblock \emph{J. ACM}, 30\penalty0 (2):\penalty0 323--342, 1983.

\bibitem[Chadha and Viswanathan(2007)]{Chadha:2007}
R.~Chadha and M.~Viswanathan.
\newblock Decidability results for well-structured transition systems with
  auxiliary storage.
\newblock In \emph{CONCUR}, pages 136--150, 2007.

\bibitem[Czerwinski et~al.(2009)Czerwinski, Fr{\"o}schle, and
  Lasota]{Czerwinski:09}
W.~Czerwinski, S.~B. Fr{\"o}schle, and S.~Lasota.
\newblock Partially-commutative context-free processes.
\newblock In \emph{CONCUR}, pages 259--273, 2009.

\bibitem[Czerwinski et~al.(2012)Czerwinski, Hofman, and
  Lasota]{Czerwinski:Concur:2012}
W.~Czerwinski, P.~Hofman, and S.~Lasota.
\newblock Reachability problem for weak multi-pushdown automata.
\newblock In \emph{CONCUR}, pages 53--68, 2012.

\bibitem[D'Osualdo et~al.(2012)D'Osualdo, Kochems, and Ong]{Soter:2012}
E.~D'Osualdo, J.~Kochems, and C.-H.~L. Ong.
\newblock Soter: an automatic safety verifier for {Erlang}.
\newblock In \emph{AGERE! '12}, pages 137--140, 2012.

\bibitem[D'Osualdo et~al.(2013)D'Osualdo, Kochems, and Ong]{Soter:2013}
E.~D'Osualdo, J.~Kochems, and C.-H.~L. Ong.
\newblock Automatic verification of {Erlang}-style concurrency.
\newblock In \emph{SAS}, 2013.
\newblock To Appear.

\bibitem[Esparza(1997)]{Esparza:1997}
J.~Esparza.
\newblock Petri nets, commutative context-free grammars, and basic parallel
  processes.
\newblock \emph{Fundam. Inform.}, 31\penalty0 (1):\penalty0 13--25, 1997.

\bibitem[Esparza and Ganty(2011)]{Esparza:2011}
J.~Esparza and P.~Ganty.
\newblock Complexity of pattern-based verification for multithreaded programs.
\newblock In \emph{POPL}, pages 499--510, 2011.

\bibitem[Esparza and Podelski(2000)]{Esparza:2000}
J.~Esparza and A.~Podelski.
\newblock Efficient algorithms for pre* and post* on interprocedural parallel
  flow graphs.
\newblock In \emph{POPL}, pages 1--11, 2000.

\bibitem[Esparza et~al.(2010)Esparza, Ganty, Kiefer, and
  Luttenberger]{Esparza:ArXiv:2010}
J.~Esparza, P.~Ganty, S.~Kiefer, and M.~Luttenberger.
\newblock Parikh's theorem: A simple and direct construction.
\newblock \emph{CoRR}, abs/1006.3825, 2010.

\bibitem[Finkel and Schnoebelen(2001)]{Finkel:01}
A.~Finkel and P.~Schnoebelen.
\newblock Well-structured transition systems everywhere!
\newblock \emph{Theoretical Computer Science}, 256\penalty0 (1-2):\penalty0
  63--92, 2001.

\bibitem[Flanagan and Qadeer(2003)]{Flanagan:2003}
C.~Flanagan and S.~Qadeer.
\newblock Thread-modular model checking.
\newblock In \emph{SPIN}, pages 213--224, 2003.

\bibitem[Ganty and Majumdar(2012)]{Ganty:2012}
P.~Ganty and R.~Majumdar.
\newblock Algorithmic verification of asynchronous programs.
\newblock \emph{ACM Trans. Program. Lang. Syst.}, 34\penalty0 (1):\penalty0 6,
  2012.

\bibitem[Ganty et~al.(2009)Ganty, Majumdar, and Rybalchenko]{Ganty:2009}
P.~Ganty, R.~Majumdar, and A.~Rybalchenko.
\newblock Verifying liveness for asynchronous programs.
\newblock In \emph{POPL}, pages 102--113, 2009.

\bibitem[Geeraerts et~al.(2012)Geeraerts, Heu{\ss}ner, and
  Raskin]{Geeraerts:2012}
G.~Geeraerts, A.~Heu{\ss}ner, and J.-F. Raskin.
\newblock Queue-dispatch asynchronous systems.
\newblock \emph{CoRR}, abs/1201.4871, 2012.

\bibitem[Henzinger et~al.(2003)Henzinger, Jhala, Majumdar, and
  Qadeer]{Henzinger:2003}
T.~A. Henzinger, R.~Jhala, R.~Majumdar, and S.~Qadeer.
\newblock Thread-modular abstraction refinement.
\newblock In \emph{CAV}, pages 262--274, 2003.

\bibitem[Heu{\ss}ner et~al.(2010)Heu{\ss}ner, Leroux, Muscholl, and
  Sutre]{Heussner:2010}
A.~Heu{\ss}ner, J.~Leroux, A.~Muscholl, and G.~Sutre.
\newblock Reachability analysis of communicating pushdown systems.
\newblock In \emph{FOSSACS}, pages 267--281, 2010.

\bibitem[Jhala and Majumdar(2007)]{Jhala:2007}
R.~Jhala and R.~Majumdar.
\newblock Interprocedural analysis of asynchronous programs.
\newblock In \emph{POPL}, pages 339--350, 2007.

\bibitem[Kahlon(2009)]{Kahlon:2009}
V.~Kahlon.
\newblock Boundedness vs. unboundedness of lock chains: Characterizing
  decidability of pairwise {CFL}-reachability for threads communicating via
  locks.
\newblock In \emph{LICS}, pages 27--36, 2009.

\bibitem[Kochems and Ong(2013)]{APCPS-long}
J.~Kochems and C.-H.~L. Ong.
\newblock Safety verification of asynchronous pushdown systems with shaped
  stacks (long version).
\newblock \url{http://www.cs.ox.ac.uk/people/jonathan.kochems/apcps.pdf}, 2013.

\bibitem[Lal and Reps(2009)]{Lal:2009}
A.~Lal and T.~W. Reps.
\newblock Reducing concurrent analysis under a context bound to sequential
  analysis.
\newblock \emph{Formal Methods in System Design}, 35\penalty0 (1):\penalty0
  73--97, 2009.

\bibitem[Long et~al.(2012)Long, Calin, Majumdar, and Meyer]{Long:2012}
Z.~Long, G.~Calin, R.~Majumdar, and R.~Meyer.
\newblock Language-theoretic abstraction refinement.
\newblock In \emph{FASE}, pages 362--376, 2012.

\bibitem[Mukund et~al.(1998{\natexlab{a}})Mukund, Kumar, Radhakrishnan, and
  Sohoni]{Mukund:Asian:1998}
M.~Mukund, K.~N. Kumar, J.~Radhakrishnan, and M.~A. Sohoni.
\newblock Towards a characterisation of finite-state message-passing systems.
\newblock In \emph{ASIAN}, pages 282--299, 1998{\natexlab{a}}.

\bibitem[Mukund et~al.(1998{\natexlab{b}})Mukund, Kumar, Radhakrishnan, and
  Sohoni]{Mukund:Icalp:1998}
M.~Mukund, K.~N. Kumar, J.~Radhakrishnan, and M.~A. Sohoni.
\newblock Robust asynchronous protocols are finite-state.
\newblock In \emph{ICALP}, pages 188--199, 1998{\natexlab{b}}.

\bibitem[Ramalingam(2000)]{Ramalingam:2000}
G.~Ramalingam.
\newblock Context-sensitive synchronization-sensitive analysis is undecidable.
\newblock \emph{ACM Trans. Program. Lang. Syst.}, 22\penalty0 (2):\penalty0
  416--430, 2000.

\bibitem[Sen and Viswanathan(2006)]{Sen:2006}
K.~Sen and M.~Viswanathan.
\newblock Model checking multithreaded programs with asynchronous atomic
  methods.
\newblock In \emph{CAV}, pages 300--314, 2006.

\bibitem[Torre and Napoli(2011)]{Torre:2011}
S.~L. Torre and M.~Napoli.
\newblock Reachability of multistack pushdown systems with scope-bounded
  matching relations.
\newblock In \emph{CONCUR}, pages 203--218, 2011.

\bibitem[Torre et~al.(2009)Torre, Madhusudan, and Parlato]{Torre:2009}
S.~L. Torre, P.~Madhusudan, and G.~Parlato.
\newblock Reducing context-bounded concurrent reachability to sequential
  reachability.
\newblock In \emph{CAV}, pages 477--492, 2009.

\bibitem[Wehrman(2006)]{Wehrman:2006}
I.~Wehrman.
\newblock Higman's theorem and the multiset order, 2006.
\newblock URL \url{{http://www.cs.utexas.edu/~iwehrman/pub/ms-wqo.pdf}}.

\end{thebibliography}
}
\iffull%
\newpage
\appendix
\section{Proof of Theorem~\ref{thm:reduction_coverability}}

\subsection{Direction: $\Leftarrow$}

We lift define a function $\M\seqpar{\cdot}$ over sequences $\NonT\ComN^*(\NComN\ComN^*)^*$ in the following way:
\begin{align*}
\M\seqpar{C_1\cdots C_n}         &= \left\{ \bigoplus_{i=1}^n \M(w_i) \mid C_i \toSF^* w_i, w_i \in \CommWords\right\}\\
\M\seqpar{C_1\cdots C_n Z\alpha} &= \M\seqpar{C_1\cdots C_n} \cdot Z \cdot \M\seqpar{\alpha}\\
\M\seqpar{X\alpha} &= X \cdot \M\seqpar{\alpha}\\
\M\seqpar{a\alpha} &= a \cdot \M\seqpar{\alpha}\\
\end{align*}
Let $U,V \subseteq \Control^\M$, we define $U \toSM V$ just if for all $\gamma' \in V$ there exists a $\gamma \in U$ such that $\gamma \toSM \gamma'$.
\begin{lemma}\label{apx:lemma:mseqpar_sim_SF}
If $\alpha \toSF \beta$ such that $\alpha \in \NonT^*$ then $\M\seqpar{\alpha} \toSM \M\seqpar{\beta}$.
\end{lemma}
\begin{proof}
Since $\alpha \toSF \beta$ we have $\alpha = X\alpha_0$ and $\beta = \alpha_1\alpha_0$ such that 
$X \rightarrow \alpha_1$. And so $\M\seqpar{\alpha} = X\M\seqpar{\alpha_0}$. We will proceed by case analysis on $X \rightarrow \alpha_1$.
\begin{itemize}
 	\item $X \rightarrow a$, $a \in \Sigma \union \{\epsilon\}$. \newline 
 		  Take $a\delta \in a\cdot\M\seqpar{\alpha_0} = \M\seqpar{a\alpha_0} = \M\seqpar{\beta}$.
 		  Then $X\delta \in \M\seqpar{\alpha}$ and $X\delta \toSM a\delta$. Hence $\M\seqpar{\alpha} \toSM \M\seqpar{\beta}$.
 	\item $X \rightarrow aA$, $a \in \Sigma$. \newline 
 		  Take $aA\delta \in a\cdot A\cdot\M\seqpar{\alpha_0} = \M\seqpar{aA\alpha_0} = \M\seqpar{\beta}$, then since
 		  $X\delta \in \M\seqpar{\alpha}$ and $X\delta \toSM aA\delta$. Hence $\M\seqpar{\alpha} \toSM \M\seqpar{\beta}$.	 
 	\item $X \rightarrow AB$, $B \in \NComN$. \newline 
 		  Suppose $AB\delta \in A \cdot B \cdot \M\seqpar{\alpha_0} = \M\seqpar{\alpha_1\alpha_0} = \M\seqpar{\beta}$, then
 		  $X\delta \in \M\seqpar{\alpha}$ and $X\delta \toSM AB\delta$. Hence $\M\seqpar{\alpha} \toSM \M\seqpar{\beta}$.	
 	\item $X \rightarrow AB$, $B \in \ComN$. \newline 
 		  Suppose ${A M'\delta \in \M\seqpar{AB\alpha_0}} = \M\seqpar{\alpha_1\alpha_0} = \M\seqpar{\beta}$. Then clearly ${M' = \M(w) \oplus M}$ such that $B \toSF^* w$ and $M\delta \in \M\seqpar{\alpha_0}$.
 		  Now then $X M\delta \in \M\seqpar{\alpha}$ and $X M\delta \toSM A(\M(w) \oplus M)\delta$.
 		  Thus $\M\seqpar{\alpha} \toSM \M\seqpar{\beta}$.	 	   
 \end{itemize} 

\end{proof}

For $U, V \subseteq \M[\Control^\M]$ define
\begin{align*}
U \parallel V &:= \{\Pi_0 \parallel \Pi_1 \mid \Pi_0 \in U, \Pi_1 \in V\}\\
\intertext{Further we define}
\M\seqpar{\Pi \,\parallel\, \Pi'} &:= \M\seqpar{\Pi} \parallel \M\seqpar{\Pi'}
\end{align*}

We say that for $U, V \subseteq \M[\Control^\M]$ $U \ChanPar \Gamma \toCM V \ChanPar \Gamma'$, just if for all $\Pi' \in V$ there exists $\Pi \in U$ such that $\Pi \ChanPar \Gamma \toCM \Pi' \ChanPar \Gamma'$. Note this means that if $\M\seqpar{\alpha} \toSM \M\seqpar{\beta}$ then clearly $\M\seqpar{\cproc{\alpha} \parallel \Pi} \ChanPar \Gamma \toCM \M\seqpar{\cproc{\beta} \parallel \Pi} \ChanPar \Gamma$ for all $\Pi$, $\Gamma$.

\begin{lemma}\label{apx:lemma:leftmost_term_simulation}
\begin{enumerate}
	\item If ${a \in \ComSigma}$ and
			  $\cproc{a\alpha\alpha'} \parallel \Pi \ChanPar \Gamma \toCF 
			   \cproc{\alpha\alpha'}  \parallel \Pi \oplus \Pi(a) \ChanPar \Gamma \oplus \Gamma(a)$
			where 
				  $\alpha \in {(\NonT \union \{\epsilon\})\ComN^*}$, 
				  $\alpha' \in (\NComN \ComN)^*$, 
			then 
			$$\M\seqpar{\cproc{a\alpha\alpha'} \parallel \Pi} \ChanPar \Gamma \toCM 
			  \M\seqpar{\cproc{\alpha\alpha'}  \parallel \Pi \oplus \Pi(a)} \ChanPar \Gamma \oplus \Gamma(a).$$
	\item If $\cproc{(\rec{c}{m})\alpha\alpha'} \parallel \Pi \ChanPar \Gamma \oplus \Gamma(\snd{c}{m}) 
			  \toCF 
			  \cproc{\alpha\alpha'}           \parallel \Pi \ChanPar \Gamma $
			where 
				  $\alpha \in {(\NonT \union \{\epsilon\})\ComN^*}$, 
				  $\alpha' \in (\NComN \ComN)^*$, 
			then 
			$$\M\seqpar{\cproc{(\rec{c}{m})\alpha\alpha'} \parallel \Pi} \ChanPar \Gamma \oplus \Gamma(\snd{c}{m})\toCM 
			  \M\seqpar{\cproc{\alpha\alpha'}             \parallel \Pi} \ChanPar \Gamma.$$
\end{enumerate}
\end{lemma}
\begin{proof}[Claim 1]
We show that 
$\M\seqpar{\cproc{a\alpha\alpha'} \parallel \Pi} \ChanPar \Gamma \toCM 
  \M\seqpar{\cproc{\alpha\alpha'} \parallel \Pi'} \ChanPar \Gamma' $
by case analysis on $a$:
\begin{itemize}
	\item $a = \snd{c}{m}$ \newline
	Then $\Pi \oplus \Pi(a) = \Pi$, 

	Take $\cproc{\gamma} \parallel \pi \in \M\seqpar{\cproc{\alpha\alpha'} \parallel \Pi}$ then 
	$\cproc{(\snd{c}{m})\gamma} \parallel \pi \in \M\seqpar{\cproc{(\snd{c}{m})\alpha\alpha'} \parallel \Pi}$. 
	Using rule \ref{def:con_mult_sem_send} we see that 
	$${\cproc{(\snd{c}{m})\gamma} \parallel \pi \ChanPar \Gamma \toCM \cproc{\gamma} \parallel \pi \ChanPar \Gamma \oplus \Gamma(\snd{c}{m})}.$$
	
	Hence we conclude 
	${\M\seqpar{\cproc{(\snd{c}{m})\alpha\alpha'} \parallel \Pi} \ChanPar \Gamma \toCM 
	     \M\seqpar{\cproc{\alpha\alpha'} \parallel \Pi} \ChanPar \oplus \Gamma(\snd{c}{m})}$.
	\item $a = \nu X$ \newline
	Then $\Pi \oplus \Pi(\nu X) = \Pi \parallel \cproc{X}$, 
	     $\Gamma = \Gamma$

	Take $\cproc{\gamma} \parallel \cproc{X} \parallel \pi \in \M\seqpar{\cproc{\alpha\alpha'} \parallel \cproc{X} \parallel\Pi} = \M\seqpar{\cproc{\alpha\alpha'} \parallel \Pi \oplus \Pi(\nu X)}$ then
	$\cproc{(\nu X)\gamma} \parallel \pi \in \M\seqpar{\cproc{(\nu X)\alpha\alpha'} \parallel \Pi}$.
	Using rule \ref{def:con_mult_sem_spawn} we see that 
	$${\cproc{(\nu X)\gamma} \parallel \pi \ChanPar \Gamma \toCM \cproc{\gamma} \parallel \cproc{X} \parallel \pi \ChanPar \Gamma}.$$ 
	
	Hence we conclude 
	${\M\seqpar{\cproc{(\nu X)\alpha\alpha'} \parallel \Pi} \ChanPar \Gamma \toCM 
	     \M\seqpar{\cproc{\alpha\alpha'} \parallel \Pi \oplus \Pi(\nu X)} \ChanPar \Gamma}$.
	\item $a = l$ \newline
	Then $\Pi    \oplus \Pi(l)    = \Pi$, 
	     $\Gamma \oplus \Gamma(l) = \Gamma$.
	Take $\cproc{\gamma} \parallel \pi \in \M\seqpar{\cproc{\alpha\alpha'} \parallel\Pi}$ then
	$\cproc{l\gamma} \parallel \pi \in \M\seqpar{\cproc{l\alpha\alpha'} \parallel \Pi}$. 
	Using rule \ref{def:con_mult_sem_label} we see that 
	$${\cproc{l\gamma} \parallel \pi \ChanPar \Gamma \toCM \cproc{\gamma} \parallel \pi \ChanPar \Gamma}.$$ 
	Hence we conclude 

	${\M\seqpar{\cproc{l\alpha\alpha'} \parallel \Pi} \ChanPar \Gamma \toCM 
	     \M\seqpar{\cproc{\alpha\alpha'} \parallel \Pi} \ChanPar \Gamma}$.
\end{itemize}
\end{proof}
\begin{proof}[Claim 2]
	Take $\cproc{\gamma} \parallel \pi \in \M\seqpar{\cproc{\alpha\alpha'} \parallel \Pi} = \M\seqpar{\cproc{\alpha\alpha'} \parallel \Pi}$
	then $\cproc{(\rec{c}{m})\gamma} \parallel \pi \in \M\seqpar{\cproc{(\rec{c}{m})\alpha\alpha'} \parallel \Pi}$

	Then using rule \ref{def:con_mult_sem_rec} we see that 
	$${\cproc{(\rec{c}{m})\gamma} \parallel \pi \ChanPar \Gamma \oplus \Gamma(\snd{c}{m}) \toCM 
	\cproc{\gamma} \parallel \pi \ChanPar \Gamma}.$$ 
	
	Hence we conclude 
	${\M\seqpar{\cproc{(\rec{c}{m})\alpha\alpha'} \parallel \Pi} \ChanPar \Gamma \oplus \Gamma(\snd{c}{m}) \toCM 
	     \M\seqpar{\cproc{\alpha\alpha'} \parallel \Pi} \ChanPar \Gamma}$.
\end{proof}
\pagebreak
\begin{lemma}\label{apx:lemma:leftmost_term_simulation_with_comsideeffects}
	\item If $a_1 \cdots a_n \in \ComSigma^*$, $\alpha_i \in \ComN^*$ and $\alpha' \in (\NComN \ComN)^*$, 
			\begin{align*}
			 \cproc{\alpha_1\alpha'} \parallel \Pi(\epsilon) \ChanPar \Gamma(\epsilon) &\toCF^*  
			 \cproc{a_1\alpha_2\alpha'} \parallel \Pi(\epsilon) \ChanPar \Gamma(\epsilon) \\
			 & \toCF \cproc{\alpha_2\alpha'} \parallel \Pi(a_1) \ChanPar \Gamma(a_1) \\
			 & \toCF^* \cdots \toCF^* \\
			 & \cproc{a_n\alpha_{n+1}\alpha'} \parallel \Pi(a_1\cdots a_{n-1}) \ChanPar \Gamma(a_1\cdots a_{n-1}) \\
			 & \toCF^* \cproc{\alpha_{n+1}\alpha'} \parallel \Pi(a_1\cdots a_{n}) \ChanPar \Gamma(a_1\cdots a_{n})
			\end{align*}  
			then 
			$$\M\seqpar{\cproc{\alpha_1\alpha'} \parallel \Pi(\epsilon)} \ChanPar \Gamma(\epsilon) \toCM^* 
			\M\seqpar{\cproc{\alpha_{n+1}\alpha'}  \parallel \Pi(a_1\cdots a_{n})} \ChanPar \Gamma(a_1\cdots a_{n}).$$
\end{lemma}
\begin{proof}
We prove the claim by induction on $n$.
For $n = 0$ the claim is vacuously true.

For $n = k + 1$, assuming the claim holds for $k$ it is enough to show that if
\begin{align*}
\cproc{\alpha_{k+1}\alpha'} \parallel \Pi(a_1\cdots a_{k}) \ChanPar \Gamma(a_1\cdots a_{k}) &\toCF^*
\cproc{a_{k+1}\alpha_{k+2}\alpha'} \parallel \Pi(a_1\cdots a_{k}) \ChanPar \Gamma(a_1\cdots a_{k})\\ 
&\toCF^* \cproc{\alpha_{k+2}\alpha'} \parallel \Pi(a_1\cdots a_{k+1}) \ChanPar \Gamma(a_1\cdots a_{k+1})
\end{align*}
then 
$$\M\seqpar{\cproc{\alpha_{k+1}\alpha'} \parallel \Pi(a_1\cdots a_{k})} \ChanPar \Gamma(a_1\cdots a_{k}) 
\toCM^*
\M\seqpar{\cproc{\alpha_{k+2}\alpha'}  \parallel \Pi(a_1\cdots a_{k+1})} \ChanPar \Gamma(a_1\cdots a_{k+1})$$
which we obtain by repeatedly applying Lemma \ref{apx:lemma:mseqpar_sim_SF} and then Lemma \ref{apx:lemma:leftmost_term_simulation}.
\end{proof}

\begin{lemma}\label{apx:lemma:term_sideeffect_pop}
If 
$\cproc{a\alpha\alpha'} \parallel \Pi \ChanPar \Gamma \toCF 
	\cproc{\alpha\alpha'} \parallel \Pi' \ChanPar \Gamma' \toCF^* 
	\cproc{\alpha'} \parallel \Pi'' \ChanPar \Gamma''$ 
where ${a \in \Sigma}$, 
	  $\alpha \in \ComN^*$, 
	  $\alpha' \in (\NComN \ComN)^*$, 
	  ${\alpha \toSF^* w \in \CommTermWords}$, 
	  $\Pi'' = \Pi' \oplus \Pi(w)$, 
	  $\Gamma'' = \Gamma' \oplus \Gamma(w)$
then 
$$\M\seqpar{\cproc{a\alpha\alpha'} \parallel \Pi} \ChanPar \Gamma \toCM \M\seqpar{\cproc{\alpha\alpha'} \parallel \Pi'} \ChanPar \Gamma' \toCM \M\seqpar{\cproc{\alpha'} \parallel \Pi''} \ChanPar \Gamma''.$$
\end{lemma}
\begin{proof}
For
 $\M\seqpar{\cproc{a\alpha\alpha'} \parallel \Pi} \ChanPar \Gamma \toCM 
  \M\seqpar{\cproc{\alpha\alpha'} \parallel \Pi'} \ChanPar \Gamma' $
we appeal to Lemma \ref{apx:lemma:leftmost_term_simulation}.
Now since $\alpha \toSF^* w$ we have 
$\cproc{\alpha\alpha'} \parallel \Pi' \ChanPar \Gamma' \toCF^* 
 \cproc{\alpha'} \parallel \Pi' \oplus \Pi(w) \ChanPar \Gamma' \oplus \Gamma(w)$.

Let $M = \M(w)$ and
take $\cproc{\gamma} \parallel \pi' \oplus \Pi(M) \ChanPar \Gamma' \oplus \Gamma(M) \in 
{\M\seqpar{\cproc{\alpha'} \parallel \Pi' \oplus \Pi(w)} \ChanPar \Gamma' \oplus \Gamma(w)}$
Then note
$\cproc{M\gamma} \parallel \pi' \in \M\seqpar{\cproc{\alpha\alpha'} \parallel \Pi'}$ 
and using rule \ref{def:con_mult_sem_disp}
$$\cproc{M\gamma} \parallel \pi' \ChanPar \Gamma' \toCM \cproc{\gamma} \parallel \pi' \oplus \Pi(M) \ChanPar \Gamma' \oplus \Gamma(M).$$
\end{proof}

\begin{lemma}\label{apx:lemma:nonterm_sideeffect_partialpop}
If 
$\cproc{a\alpha\alpha'}   \parallel \Pi \ChanPar \Gamma \toCF 
	\cproc{\alpha\alpha'} \parallel \Pi' \ChanPar \Gamma' \toCF^* 
	\cproc{w\alpha'}      \parallel \Pi' \ChanPar \Gamma' \toCF^* 
	\cproc{w'\alpha'}     \parallel \Pi'' \ChanPar \Gamma''$ 
where ${a \in \Sigma}$, 
	  $\alpha \in \ComN^*$, 
	  $\alpha' \in (\NComN \ComN)^*$, 
	  $\alpha'' \in \NComN^*$ and
	  ${\alpha \toSF^* w \in \CommWords}$, 
	  $w' \in \CommNonTermWords$,
	  $\Pi'' = \Pi' \oplus \Pi(w)$, 
	  $\Gamma'' = \Gamma' \oplus \Gamma(w)$
then 
$$\M\seqpar{\cproc{a\alpha\alpha'} \parallel \Pi} \ChanPar \Gamma \toCM \M\seqpar{\cproc{\alpha\alpha'} \parallel \Pi'} \ChanPar \Gamma' \toCM  \cproc{\M(w') \cdot \M\seqpar{\alpha'}} \parallel \M\seqpar{\Pi''} \ChanPar \Gamma''.$$
\end{lemma}
\begin{proof}
For
 $\M\seqpar{\cproc{a\alpha\alpha'} \parallel \Pi} \ChanPar \Gamma \toCM 
  \M\seqpar{\cproc{\alpha\alpha'} \parallel \Pi'} \ChanPar \Gamma' $
we appeal to Lemma \ref{apx:lemma:leftmost_term_simulation}.
Now since $\alpha \toSF^* w$ we have 
$\cproc{\alpha\alpha'} \parallel \Pi' \ChanPar \Gamma' \toCF^* 
\cproc{w'\alpha'} \parallel \Pi' \oplus \Pi(w) \ChanPar \Gamma' \oplus \Gamma(w)$.

Let $M' = M(w')$ then

$\cproc{M'\gamma} \parallel \pi' \oplus \Pi(M) \ChanPar \Gamma' \oplus \Gamma(M) \in \cproc{\M(w') \cdot \M\seqpar{\alpha'}} \parallel \M\seqpar{\Pi' \oplus \Pi(w)} \ChanPar \Gamma' \oplus \Gamma(w)$
and
$\cproc{M\gamma} \parallel \pi' \in \M\seqpar{\cproc{\alpha\alpha'} \parallel \Pi'}$ such that $M = \M(w)$.
Using rule \ref{def:con_mult_sem_non-term}
$$\cproc{M\gamma} \parallel \pi' \ChanPar \Gamma' \toCM \cproc{M'\gamma} \parallel \pi' \oplus \Pi(M) \ChanPar \Gamma' \oplus \Gamma(M).$$
\end{proof}

\begin{lemma}\label{apx:lemma:sim_fromcom_pop_or_nonterm}
Let $X \in \NonT$, $\alpha \in (\NComN\ComN^*)^*$, $w \in \ComSigma$ and $\beta,\alpha' \in \ComN^*$
\begin{enumerate}
	\item If $\cproc{X\beta\alpha} \ChanPar \Gamma(\epsilon) \Gamma(\epsilon) \toCF^* \cproc{\alpha} \parallel \Pi(w) \ChanPar \Gamma(w)$
	      then 
	      ${\M\seqpar{\cproc{X\beta\alpha}} \ChanPar \Gamma(\epsilon) \toCM^* \M\seqpar{\cproc{\alpha} \parallel \Pi(w)} \ChanPar \Gamma(w)}$
	\item If $\cproc{X\beta\alpha} \ChanPar \Gamma(\epsilon) \toCM^* \cproc{\alpha'\alpha} \parallel \Pi(w) \ChanPar \Gamma(w) $, then $\M\seqpar{\cproc{X\beta\alpha}} \ChanPar \Gamma(\epsilon) \toCF^* {\cproc{\M(\alpha')\cdot\M\seqpar{\alpha}}} \parallel \M\seqpar{\Pi(w)} \ChanPar \Gamma(w)$
\end{enumerate}
\end{lemma}
\begin{proof}[Claim 1]
Then $X\beta\alpha \toSF^* a\alpha_0\alpha$ where 
$a \in \ComSigma \union \{\epsilon\}$ such that $a\alpha_0 \toSF^* w$ 
so by Lemma \ref{apx:lemma:mseqpar_sim_SF} $\M\seqpar{X\beta\alpha} \toSM^* \M\seqpar{a\alpha_0\alpha}$.
By Lemma \ref{apx:lemma:term_sideeffect_pop}
$\M\seqpar{\cproc{a\alpha_0\alpha}} \ChanPar \Gamma(\epsilon) \toCM^* \M\seqpar{\cproc{\alpha} \parallel \Pi(w)} \ChanPar \Gamma(w)$ and clearly also $\M\seqpar{\cproc{X\beta\alpha}} \ChanPar \Gamma(\epsilon) \toCM^* \M\seqpar{\cproc{a\alpha_0\alpha}} \parallel \Pi(w) \ChanPar \Gamma(w)$.
\end{proof}
\begin{proof}[Claim 2]
Then $X\beta\alpha \toSF^* a\alpha_0\alpha$ where 
$a \in \ComSigma \union \{\epsilon\}$ such that $a\alpha_0 \toSF^* w\alpha'$ 
so by Lemma \ref{apx:lemma:mseqpar_sim_SF} $\M\seqpar{X\beta\alpha} \toSM^* \M\seqpar{a\alpha_0\alpha}$.
By Lemma \ref{apx:lemma:nonterm_sideeffect_partialpop}
$\M\seqpar{\cproc{a\alpha_0\alpha}} \ChanPar \Gamma(\epsilon) \toCM^* \M\seqpar{\cproc{\alpha'\alpha} \parallel \Pi(w)} \ChanPar \Gamma(w)$ and clearly also $\M\seqpar{\cproc{X\alpha}} \ChanPar \Gamma(\epsilon) \toCM^* \M\seqpar{\cproc{a\alpha_0\alpha}} \parallel \Pi(w) \ChanPar \Gamma(w)$.
\end{proof}

\begin{lemma}\label{apx:lemma:sim_fromncom_push_or_nonterm}
Let $X \in \NComN$, $\beta,\beta' \in \ComN^*$ and $\alpha,\alpha' \in (\NComN\ComN^*)^*$.
\begin{enumerate}
	\item If $\cproc{X\beta\alpha} \ChanPar \Gamma(\epsilon) \toCF^*  
	          \cproc{\alpha'\alpha} \parallel \Pi(w) \ChanPar \Gamma(w)$
	      then 
	      $\M\seqpar{\cproc{X\beta\alpha}} \ChanPar \Gamma(\epsilon) \toCM^* 
	      \M\seqpar{\cproc{\alpha'\alpha} \parallel \Pi(w)} \ChanPar \Gamma(w)$
	\item If $\cproc{X\beta\alpha} \ChanPar \Gamma(\epsilon) \toCF^* \cproc{(\rec{c}{m})\beta'\alpha'\alpha} \parallel \Pi(w) \ChanPar \Gamma(w)$, then $\M\seqpar{\cproc{X\alpha} \parallel \Pi(w)} \ChanPar \Gamma(\epsilon) \toCM^* \M\seqpar{\cproc{(\rec{c}{m})\beta'\alpha'\alpha}} \ChanPar \Gamma(\epsilon)$.
\end{enumerate}
\end{lemma}
\begin{proof}[Claim 1]
Then 
$$\cproc{X\beta\alpha} \ChanPar \Gamma(\epsilon) \toCF^* 
\cproc{X'\alpha'\alpha} \parallel \Pi(w_0) \ChanPar \Gamma(w_0)$$
such that $X' \toSF^* a\alpha_0$ where 
$a \in \ComSigma \union \{\epsilon\}$, $\alpha_0 \in \ComN^*$ such that $a\alpha_0 \toSF^* w_1$ and $w = w_0w_1$. 
By Lemma \ref{apx:lemma:leftmost_term_simulation_with_comsideeffects} 
$$\M\seqpar{\cproc{X\beta\alpha}} \ChanPar \Gamma(\epsilon) \toCM^* \M\seqpar{\cproc{X'\alpha'\alpha} \parallel \Pi(w_0)} \ChanPar \Gamma(w_0)$$
Then the proof of Lemma \ref{apx:lemma:sim_fromcom_pop_or_nonterm} Claim 1 applies to give the result.
\end{proof}
\begin{proof}[Claim 2]
Then 
$$\cproc{X\beta\alpha} \ChanPar \Gamma(\epsilon) \toCM^* \cproc{X'X''\beta'\alpha'\alpha} \parallel \Pi(w_0) \ChanPar \Gamma(w_0)$$
such that $X'' \toSF^* \rec{c}{m}$, and $X' \toSF^* w_1$ where $w=w_0w_1$. Hence
$$\cproc{X'X''\beta'\alpha'\alpha} \parallel \Pi(w_0) \ChanPar \Gamma(w_0) \toCM^* \cproc{(\rec{c}{m})\beta'\alpha'\alpha} \parallel \Pi(w) \ChanPar \Gamma(w).$$

so by Lemma \ref{apx:lemma:leftmost_term_simulation_with_comsideeffects} and Lemma \ref{apx:lemma:mseqpar_sim_SF}.
$$\M\seqpar{\cproc{X\beta\alpha}} \ChanPar \Gamma(\epsilon) \toCM^* 
  \M\seqpar{\cproc{(\rec{c}{m})\beta'\alpha'\alpha} \parallel \Pi(w)} \ChanPar \Gamma(w).$$
\end{proof}

For $\alpha_1,\ldots,\alpha_m \in \ComN^*$ and $Z_1,\ldots, Z_{m-1} \in \NComN$ define
\begin{align*}
\seqM(\alpha_1 Z_1\cdots \alpha_{m-1}Z_{m-1}\alpha_m) &:= \M(\alpha_1)Z_1 \cdots \M(\alpha_{m-1})Z_{m-1}\M(\alpha_m)\\
\seqM(\Pi \parallel \Pi') &:= \seqM(\Pi) \parallel \seqM(\Pi')
\end{align*}

\begin{proposition}\label{apx:prop:conc_reduction_simulation}
If $\cproc{S} \ChanPar \Gamma(\epsilon) \toCF^* \Pi' \ChanPar \Gamma'$
then $\seqM(\cproc{S}) \ChanPar \Gamma(\epsilon) \toCM^* \seqM(\Pi') \ChanPar \Gamma'$
\end{proposition}
\begin{proof}
Let $\Pi^f \in \M[\Control]$ and define the set
$P_{\Pi^f} = \{ \alpha \mid \exists \Pi. \cproc{\alpha} \parallel \Pi = \Pi^f \}$.
further define the set of configurations
$P := \{ \Pi \ChanPar \Gamma \mid \forall \cproc{\alpha} \in \Pi, \alpha \in \NonT(\NComN\ComN^*)^* \union \NComSigma\ComN^*(\NComN\ComN^*)^* \union \NonT \union P_{\Pi^f}\}$.

Now suppose that for some $\Pi, \Pi'$ and $\Gamma, \Gamma'$
$$\Pi \ChanPar \Gamma := \Pi_0 \ChanPar \Gamma_0 \toCF^* 
  \Pi_1 \ChanPar \Gamma_1 \toCF^* \cdots \toCF^* \Pi_n \ChanPar \Gamma_n =: \Pi' \ChanPar \Gamma'$$ 
such that $\Pi_i \ChanPar \Gamma_i \in P$ for $i = 0,\ldots n$. Without loss of generality we can assume that for all $i = 0,...,n$, $\Pi_i = \Pi_i^a \parallel \Pi_i^f$ such that for all 
$\cproc{\alpha} \in \Pi_i^f$ we have $\cproc{\alpha} \in P_{\Pi^f}$ and 
$\cproc{\alpha}$ is not involved in any transitions in 
$\Pi_i \ChanPar \Gamma_i \toCF^* \Pi_n \ChanPar \Gamma_n$. 
Note that we are not loosing generality, since a reduction 
$\cproc{\alpha} \parallel \Pi \toCF^* \cproc{\alpha} \parallel \Pi'$ can either be pre-empted or goes through a process state in $\NComSigma\ComN^*(\NComN\ComN^*)^*$. 
Note this also means that $\Pi_{i+1}^f = \Pi_{i}^f \parallel {\Pi'}_{i}^f$. We further assume w.l.o.g that for each $i$ it is the case that 
$\Pi^a_i = \cproc{\alpha} \parallel \Pi'_i$ and $\Pi^a_{i+1} = \cproc{\alpha'} \parallel \Pi'_i \oplus \Pi(w)$ 
and $\Gamma_{k+1} \oplus \Gamma(w') = \Gamma_k \oplus \Gamma(w)$ 
for some $w \in \ComSigma^*$ and $w' \in \{\epsilon\} \union \ComSigma$,
i.e. during each $\Pi^a_i \ChanPar \Gamma_i \toCF^* \Pi^a_{i+1} \ChanPar \Gamma_{i+1}$ only one process makes progress (note this can be achieved by delaying receptions and performing sends and spawns as early as possible) and none of the intermediate steps are configurations of $P$.

We will prove by induction on $n$: 
$$\M\seqpar{\Pi^a_0} \parallel \tilde{\Pi}^f_0 \ChanPar \Gamma_0 \toCM^* \M\seqpar{\Pi^a_1} \parallel \tilde{\Pi}^f_1 \ChanPar \Gamma_1 \toCM^* \cdots \toCM^* \M\seqpar{\Pi_n} \parallel \tilde{\Pi}^f_n \ChanPar \Gamma_n$$
where for all $i = 0,\ldots n$ and $\cproc{\alpha} \in \Pi_i^f$ we have either $\Pi_i^f(\cproc{\alpha}) = \tilde{\Pi}^f_i(\M\seqpar{\cproc{\alpha}})$ or $\Pi_i^f(\cproc{\alpha}) = \tilde{\Pi}^f_i(\cproc{\M(\alpha_0) \cdot \M\seqpar{\alpha_1}})$, $\alpha = \alpha_0\alpha_1$.
\begin{itemize}
	\item $n = 0$. \newline
	The claim holds trivially.
	\item $n = k+1$, assuming the claim holds for $k$. \newline
	To prove the inductive claim we need to show that from
	$\Pi_k = \cproc{\alpha} \parallel \Pi'_k$, $\Pi_{k+1} =\cproc{\alpha'} \parallel \Pi'_k \oplus \Pi(w)$
	and
	$\Gamma_{k+1} \oplus \Gamma(w') = \Gamma_k \oplus \Gamma{w}$
	where $\Pi_k \ChanPar \Gamma_k \toCF^* \Pi_{k+1} \ChanPar \Gamma_{k+1}$,
	we can infer $\M\seqpar{\Pi_k} \ChanPar \Gamma_k \toCF^* \M\seqpar{\Pi_{k+1}} \ChanPar \Gamma_{k+1}$.
	We will do so by a case analysis on the shape of $\alpha$ and $\alpha'$.
	\begin{itemize}
		\item $\alpha, \alpha' \in (\NComN\ComN^*)^*$ \newline
		Then $\alpha = X\alpha_0\alpha_1$, $X \in \NComN$, $\alpha_0 \in \ComN^*$, 
		$\alpha_1 \in (\NComN\ComN^*)^*$ and $\alpha' = \alpha'_0\alpha_1$ where $\alpha'_0 \in \epsilon \union (\NComN\ComN^*)^*$, i.e. either we increase the call-stack or we pop one non-commutative non-terminal off the call-stack. Otherwise we would end up either in an intermediate configuration in $P$ or in a different case.
		\begin{itemize}
			\item Case $\alpha'_0 = \epsilon$.\newline
			Then $\cproc{X\alpha_0\alpha_1} \parallel \Pi'_k \ChanPar \Gamma_k \toCF^* 
				\cproc{X'\alpha_2\alpha_0\alpha_1} \parallel \Pi'_k \ChanPar \Gamma_k$ such that 
				${X' \in \ComN}$, $\alpha_2 \in \ComN^*$ and 
				$\cproc{X'\alpha_2\alpha_0\alpha_1} \parallel \Pi'_k \ChanPar \Gamma_k \toCF^*
					\cproc{w\alpha_1} \parallel \Pi'_k \ChanPar \Gamma_k \toCF^* 
					\cproc{\alpha_1} \parallel \Pi'_k \oplus \Pi(w) \ChanPar \Gamma_k \oplus \Gamma(w)$, where $w \in \ComSigma^*$ such that 
					$\Pi_{k+1} =\cproc{\alpha'} \parallel \Pi'_k \oplus \Pi(w)$
					and
					$\Gamma_{k+1} = \Gamma_k \oplus \Gamma{w}$.
					Lemma \ref{apx:lemma:leftmost_term_simulation} then allows us to conclude that
					$\M\seqpar{\cproc{X\alpha_0\alpha_1} \parallel \Pi'_k} \ChanPar \Gamma_k 
					\toCM^* 
				    \M\seqpar{\cproc{X'\alpha_2\alpha_0\alpha_1} \parallel \Pi'_k} \ChanPar \Gamma_k$ and
				    Lemma \ref{apx:lemma:sim_fromcom_pop_or_nonterm}.1 gives us
				    $\M\seqpar{\cproc{X'\alpha_2\alpha_0\alpha_1} \parallel \Pi'_k} \ChanPar \Gamma_k \toCM^*
				    \M\seqpar{\cproc{\alpha_1} \parallel \Pi'_k \oplus \Pi(w)} \ChanPar \Gamma_k \oplus \Gamma(w) = \Pi_{k+1} \ChanPar \Gamma_{k+1}$.
			\item Case $\alpha'_0 \neq \epsilon$.\newline
				Follows directly from Lemma \ref{apx:lemma:sim_fromncom_push_or_nonterm}.1
		\end{itemize}
		\item $\alpha \in \NComSigma\ComN^*(\NComN\ComN^*)^*$ and $\alpha' \in (\NComN\ComN^*)^*$ \newline
		Follows from Lemma \ref{apx:lemma:term_sideeffect_pop}.
		\item $\alpha \in (\NComN\ComN^*)^*$ and $\alpha' \in \NComSigma\ComN^*(\NComN\ComN^*)^*$ \newline
		Follows from Lemma \ref{apx:lemma:sim_fromncom_push_or_nonterm}.2
		\item $\alpha \in \NonT$ and $\alpha' \in (\NComN\ComN^*)^*$\newline
		We can assume that $\alpha \in \ComN$ since otherwise a case above already applies.
		By the definition of $\ComN$ we can thus infer that $\alpha' = \epsilon$ since otherwise 
		$\alpha$ would not be commutative.
		Thus Lemma \ref{apx:lemma:sim_fromcom_pop_or_nonterm}.1 applies.
		\item $\alpha \in \NonT$ and $\alpha' \in \NComSigma\ComN^*(\NComN\ComN^*)^*$\newline
		There is nothing to prove for this case as, similarly to the case above, either $\alpha \in \NComN$ and so a case above applies or $\alpha \in \ComN$ but then $\alpha' \notin \NComSigma\ComN^*(\NComN\ComN^*)^*$ which is impossible; so the former must be the case.
		\item $\alpha \in (\NComN\ComN^*)^*$ and $\alpha' \in P_{\Pi_f}$\newline
			If $\alpha' \in (\NComN\ComN^*)^* \union \NComSigma\ComN^*(\NComN\ComN^*)^*$ the above cases apply. Otherwise it must be the case that $\alpha' \in \ComN^*(\NComN\ComN^*)^* \union \ComSigma\ComN^*(\NComN\ComN^*)^*$.
			\begin{itemize}
				\item $\alpha' \in \ComN^*(\NComN\ComN^*)^*$ \newline
				So it must be the case that
				$\alpha = X\alpha_0\alpha_1$, $X \in \NComN$, $\alpha_0 \in \ComN^*$ $\alpha_1 \in (\NComN\ComN^*)^*$ and $\alpha' = \alpha'_0\alpha'_1\alpha_1$ where 
				$\alpha'_1 \in (\NComN\ComN^*)^*$, $\alpha'_0 \in \ComN^*$
				Lemma \ref{apx:lemma:sim_fromcom_pop_or_nonterm}.2 applies to give
				$$\M\seqpar{\cproc{X\alpha_0\alpha_1} \parallel \Pi_k} \ChanPar \Gamma_k \toCM^* \cproc{\M(\alpha'_0) \cdot \M\seqpar{\alpha'_1\alpha_1}}\parallel \M\seqpar{\Pi_k} \parallel \M\seqpar{\Pi(w)} \ChanPar \Gamma_k \oplus \Gamma(w)$$

				\item $\alpha' \in \ComSigma\ComN^*(\NComN\ComN^*)^*$\newline
				Follows from Lemma \ref{apx:lemma:leftmost_term_simulation_with_comsideeffects}
			\end{itemize}
		\item $\alpha \in \NComSigma\ComN^*(\NComN\ComN^*)^*$ and $\alpha' \in P_{\Pi_f}$\newline
			Unless $\alpha' \in \ComN^*(\NComN\ComN^*)^* \union \ComSigma\ComN^*(\NComN\ComN^*)^*$ this case is covered by a case above.
			The remaining follows from Lemma \ref{apx:lemma:nonterm_sideeffect_partialpop}.
		\item $\alpha \in \NonT$ and $\alpha' \in P_{\Pi_f}$\newline
			Unless $\alpha' \in \ComN^*(\NComN\ComN^*)^* \union \ComSigma\ComN^*(\NComN\ComN^*)^*$ this case is covered by a case above.
			The remaining follows from Lemma \ref{apx:lemma:mseqpar_sim_SF} and Lemma \ref{apx:lemma:leftmost_term_simulation_with_comsideeffects}
	\end{itemize}
	This concludes the proof of the inductive step.
\end{itemize}
Now we apply the above for the case that
$\Pi_0 \ChanPar \Gamma_0 = \cproc{S} \ChanPar \Gamma(\epsilon)$ and $\Pi^f := \Pi'$.
We can then see that $\M\seqpar{\cproc{S}} \ChanPar \Gamma(\epsilon) \toCM^* \tilde{\Pi}^f \ChanPar \Gamma'$.

Then since for all $\alpha \in \ComN^*(\NComN\ComN^*)^*$ it is the case that 
$\seqM(\alpha) \in \M\seqpar{\alpha}$ and further for 
$\alpha_0 \in \ComN^*$, $\alpha_1 (\NComN\ComN^*)^*$
$\seqM(\alpha_0\alpha_1) \in \M(\alpha_0) \cdot \M\seqpar{\alpha_1}$ we can deduce from the definition of \toCM on sets of configurations that
$$\seqM(\cproc{S}) \toCM^* \seqM(\Pi') \ChanPar \Gamma'$$ which concludes the proof.
\end{proof}

\subsection{Direction: $\Rightarrow$}


\begin{align*}
\ceil{M} =& \left\{ \left\{C_i \mapsto \sum_{(C_i,w) \in M'} M'(C_i,w)\right\} \,\left|\,
	\begin{aligned}
	&M = \bigoplus_{(C_i,w) \in M'}\bigoplus_{j=1}^{M'(C_i,w)} \M(w),\\ 
	&M' \in \M[\{(C,w) \mid C\toSF^* w, C \in \NonT\}] 	
	\end{aligned} 
	\right.\right\}\\
\sem{M} =& \{\alpha \in \Control \mid \M(\alpha) = \ceil{M}\}\\
\sem{\gamma_1\cdots \gamma_n} =& \sem{\gamma_1}\cdots\sem{\gamma_n} \text{ where } \gamma_i \in \NonT \union \Cache \union \Sigma
\end{align*}

Define for $V, W \subseteq \M[\Control]$
\begin{align*}
V \parallel W &= \{\Pi \parallel \Pi' \mid \Pi \in V, \Pi' \in W\}\\
\intertext{and $\Pi,\Pi' \in \M[\Control^\M]$}
\sem{\Pi \parallel \Pi'} &= \sem{\Pi} \parallel \sem{\Pi'}
\end{align*}

\begin{definition}[Simulation Relation]
Let $R \subseteq S \times S'$ where $(S,\rightarrow_S), {(S', \rightarrow_{S'})}$ are transition systems. We say R is a \emph{weak $(S,S')$-simulation} just if
$$(q,p) \in R \text{ and } q \rightarrow_S q' \implies p \rightarrow^*_{S'} p' \text{ and } (q',p') \in R.$$
\end{definition}

Let $\simulatedS \is \{ (\gamma,\alpha) \mid \alpha \in \sem{\gamma}\} \subseteq \Control^\M \times \Control$ and $\simulatedC \is \{ (\Pi_1 \ChanPar \Gamma, \Pi_2 \ChanPar \Gamma) \mid \Pi_2 \in \sem{\Pi_1} \}$.

\begin{proposition}[Sequential Simulation]\label{apx:prop:sequential_simulation}
$\simulatedS$ is a weak simulation relation.
\end{proposition}
\begin{proof}
Let $\gamma \in \Control^M$ and $\alpha \in \Control$ such that $\alpha \in \sem{\gamma}$ and $\gamma \toSM \gamma'$.

Since $\gamma \toSM \gamma'$ we know that $\gamma = XM\gamma_0$ and $\gamma' = \gamma_1\gamma_0$.
Hence by definition of $\sem{\gamma}$ it must be that $\alpha = X\beta\alpha_0$ such that $\beta \in \sem{M}$ and $\alpha_0 \in \sem{\gamma_0}$.

We will prove that there exists a $\alpha' \simulatesS \gamma'$ such that $\alpha \toSF \alpha'$ by case analysis on the type of rule used for $\gamma \toSM \gamma'$.
\begin{enumerate}
	\item[Claim 1 and 2:] $X \rightarrow a \in G$, $a \in \Sigma \union \{\epsilon\}$. 
			\emph{(trivial)}\newline
			Clearly $X\beta\alpha_0 \toSF a\beta\alpha_0$, and $\gamma_1 = aM$. Clearly $a\beta\alpha_0 =: \alpha' \in \sem{aM\gamma_0} = \sem{\gamma'}$ and so $\alpha' \simulatesS \gamma'$. 
	\item[Claim 3:] $X \rightarrow aA \in G$. 
			\emph{(trivial)}\newline
			Then $X\beta\alpha_0 \toSF aA\beta\alpha_0$ and $\gamma_1 = aAM$.
			Clearly $aA\beta\alpha_0 =: \alpha' \in \sem{aAM\gamma_0} = \sem{\gamma'}$ and so $\gamma' \simulatedS \alpha'$.
	\item[Claim 4:] $X \rightarrow AB \in G$, $B \in \ComN$, $B \toSF^* w$.
			\emph{(non-trivial)}\newline
			Then $X\beta\alpha_0 \toSF AB\beta\alpha_0$. 
			To prove $A\, B\, \beta\alpha_0 =: \alpha' \in \sem{A(\M(w) \oplus M)\gamma_0} =\sem{\gamma'}$ we need to show that 
			$B\, \beta \in \sem{\M(w) \oplus M}$.
			Since $\beta \in \sem{M}$ we know that $\M(\beta) = \ceil{M}$. 
			It remains to prove $\M(B) \oplus \M(\beta) \in \ceil{\M(w) \oplus M}$.
			Since ${\M(\beta) \in \ceil{M}}$ there exists $M' \in \M[\{(C,w) \mid C\toSF^* w, C \in \NonT\}]$ such that 

			$$\M(\beta) = \left\{C_i \mapsto \sum_{(C_i,w) \in M'} M'(C_i,w)\right\}$$
			
			and $M = \bigoplus_{(C_i,w) \in M'}\bigoplus_{j=1}^{M'(C_i,w)} \M(w)$.
			Then writing $M'_B := M'\oplus \M((B,w))$ it is the case that $M'_B \in \M[\{(C,w) \mid C\toSF^* w, C \in \NonT\}]$ and 
			\begin{align*}
			M \oplus \M(w) &= \left(\bigoplus_{(C_i,w') \in M'}\bigoplus_{j=1}^{M'(C_i,w')} \M(w')\right) \oplus \M(w)\\
							&= \bigoplus_{(C_i,w') \in M'_B}\bigoplus_{j=1}^{M'_B(C_i,w')} \M(w')
			\end{align*}
 			Thus we can conclude that
 			$$f_B := \left\{C_i \mapsto \sum_{(C_i,w) \in M'_B} M'_B(C_i,w)\right\} \in \ceil{\M(w) \oplus M}$$
 			and since $M'_B(B,w) = M'(B,w) + 1$ and $M'_B(B',w') = M'(B',w')$ if either $B \neq B'$ or $w \neq w'$, it is the case that $f_B(B) = \M(\beta)(B) + 1$ and $f_B(C) = \M(\gamma)(C)$ if $C \neq B$. Hence $f_B = \M(B) \oplus \M(\beta)$ and so $\M(B) \oplus \M(\beta) \in \ceil{\M(w) \oplus M}$ which implies $\alpha' \simulatesS \gamma'$ and concludes the proof of this case.
	\item[Claim 5:] $X \rightarrow AB \in G$, $B \in \NComN$.
		  \emph{(trivial)}\newline
		  Then $X\beta\alpha_0 \toSF AB\beta\alpha_0$ and $\gamma_1 = ABM$.
		  Clearly $AB\beta\alpha_0 = \alpha' \in \sem{ABM\gamma_0} = \sem{\gamma'}$ and so $\alpha' \simulatesS \gamma'$.
\end{enumerate}
The claim holds for all cases which concludes the proof.
\end{proof}

\begin{proposition}[Concurrent Simulation]\label{apx:prop:concurrent_simulation}
$\simulatedC$ is a weak simulation relation.
\end{proposition}

\begin{proof}
Let $\Pi_1 \in \M[\Control^\M]$, $\Pi_2 \in \M[\Control]$ and $\Gamma, \Gamma' \in \Chan \rightarrow \M[\Msg]$ such that $\Pi_1 \ChanPar \Gamma \simulatedC \Pi'_1 \ChanPar \Gamma$ and suppose that $\Pi_1 \ChanPar \Gamma \toCM \Pi'_1 \ChanPar \Gamma'$.

We will prove that there exists a $\Pi'_2$ such that $\Pi'_1 \ChanPar \Gamma' \simulatedC \Pi'_2 \ChanPar \Gamma'$ and $\Pi_2 \ChanPar \Gamma \toCF^* \Pi'_2 \ChanPar \Gamma'$ by case analysis on the rule used for $\Pi_1 \ChanPar \Gamma \toCM \Pi'_1 \ChanPar \Gamma'$.
\begin{itemize}
	\item Rule \ref{def:con_mult_sem_interleave} \newline
	Follows immediately by Proposition \ref{apx:prop:sequential_simulation}.
	\item Rule \ref{def:con_mult_sem_rec} \newline
	Then $\Pi_1 = \cproc{(\rec{c}{m})\, \gamma} \parallel \Pi^0_1$ and 
	$\Gamma = \Gamma' \oplus \Gamma(\snd{c}{m})$
	and $\Pi'_1 = \cproc{\gamma}\parallel \Pi^0_1$.
    Hence $\Pi_2 = \cproc{(\rec{c}{m})\alpha} \parallel \Pi^0_2 \in \sem{\Pi_1}$ with 
    $\alpha \in \sem{\gamma}$ and so
    $\cproc{\alpha}\parallel \Pi^0_2 \in \sem{\Pi'_1}$.
    And using rule \ref{def:conc_faith_rec}
    $\cproc{(\rec{c}{m})\, \alpha} \parallel \Pi^0_2 \ChanPar  \Gamma' \oplus \Gamma(\snd{c}{m}) \toCF 
    \cproc{\alpha}\parallel \Pi^0_2 \ChanPar  \Gamma'$.
	\item Rule \ref{def:con_mult_sem_spawn} \newline
	Then $\Pi_1 = \cproc{(\nu X)\, \gamma} \parallel \Pi^0_1$,
    $\Pi'_1 = \cproc{\gamma} \parallel \cproc{X} \parallel \Pi^0_1$ and $\Gamma' = \Gamma$.
    Hence $\Pi_2 = \cproc{(\nu X)\, \alpha} \parallel \Pi^0_2 \in \sem{\Pi_1}$ with
    $\alpha \in \sem{\gamma}$ and
    so $\cproc{\alpha} \parallel \cproc{X} \parallel \Pi^0_2  \in \sem{\Pi'_1}$.
    By rule \ref{def:conc_faith_spawn}
    $\cproc{(\nu X)\, \alpha} \parallel \Pi^0_2 \ChanPar \Gamma \toCF
     \cproc{\alpha} \parallel \cproc{X} \parallel \Pi^0_2 \ChanPar \Gamma$.
	\item Rule \ref{def:con_mult_sem_send} \newline
	Then $\Pi_1 = \cproc{(\snd{c_j}{m})\, \gamma} \parallel \Pi^0_1$,
	$\Pi'_1 = \cproc{\gamma} \parallel \Pi^0_1$ and $\Gamma' = \Gamma \oplus \Gamma(\snd{c}{m})$.
	Then $\Pi_2 = \cproc{(\snd{c_j}{m})\, \alpha} \parallel \Pi^0_2 \in \sem{\Pi_1}$ with  
	$\alpha \in \sem{\gamma}$ and so
	$\cproc{\alpha} \parallel \Pi^0_2 \in \sem{\Pi'_1}$.
	By rule \ref{def:con_mult_sem_send} we can see
	$\cproc{(\snd{c_j}{m})\, \alpha} \parallel \Pi^0_2 \ChanPar \Gamma \toCF
	 \cproc{\alpha} \parallel \Pi^0_2 \ChanPar  \Gamma \oplus \Gamma(\snd{c}{m})$.
	\item Rule \ref{def:con_mult_sem_label}\newline
	Then $\Pi_1 = \cproc{l\, \gamma} \parallel \Pi^0_1$,  
	$\Pi'_1 = \cproc{\gamma} \parallel \Pi^0_1$ and $\Gamma' = \Gamma$.
	Then $\Pi_2 = \cproc{l\, \alpha} \parallel \Pi^0_2 \in \sem{\Pi_1}$ with
	$\alpha \in \sem{\gamma}$ and
	$\cproc{\alpha} \parallel \Pi^0_2 \in \sem{\Pi'_1}$.
	By rule \ref{def:conc_faith_label}  
	$\cproc{l\, \alpha} \parallel \Pi^0_2 \ChanPar  \Gamma \toCF
	 \cproc{\alpha} \parallel \Pi^0_2 \ChanPar  \Gamma$.
	\item Rule \ref{def:con_mult_sem_disp} \newline
	Then $\Pi_1 = \cproc{M\, X \,\gamma} \parallel \Pi^0_1$
	such that $M \in \TermCache$, $\Gamma' = \Gamma \oplus \Gamma(M)$, 
	$X \in \NComN$
	and $\Pi'_1 = \cproc{X\, \gamma}  \parallel \Pi^0_1 \parallel \Pi(M)$.
	Also $\Pi_2 = \cproc{\beta\, X \,\alpha} \parallel \Pi^0_2 \in \sem{\Pi_1}$ and
	$\beta\, X\, \alpha \in \sem{M\,X\,\gamma}$.

	Hence $X\, \alpha \in \sem{X\,\gamma}$ and $\beta \in \sem{M}$
	and hence $\M(\beta) = \ceil{M}$. Thus $\beta \toSF^* w$, $w \in \ComSigma^*$ such that 
	$\M(w) = M$. 
	Hence using rules \ref{def:conc_faith_send}, \ref{def:conc_faith_spawn} and 
	\ref{def:conc_faith_inter} repeatedly we can see that 
	$\cproc{\beta\, X \,\alpha} \parallel \Pi^0_2  \ChanPar \Gamma \toCF^* 
	 \cproc{X \,\alpha} \parallel \Pi^0_2 \parallel \Pi(w) \ChanPar \Gamma \oplus \Gamma(w) = 
	 \cproc{X \,\alpha} \parallel \Pi^0_2 \parallel \Pi(M) \ChanPar \Gamma \oplus \Gamma(M) =: \Pi'_2
	 $.
	and $\Pi'_2 \in \sem{\Pi'_1}$.
	\item Rule \ref{def:con_mult_sem_non-term} \newline
	Then $\Pi_1 = \cproc{M\, X \,\gamma}  \parallel \Pi^0_1$
	and $\Pi'_1 = \cproc{M'\,X\, \gamma}  \parallel \Pi^0_1 \parallel \Pi(M)$
	such that $m \in \MixedCache$, $M' \in \NonTermCache$, $\Gamma' = \Gamma \oplus \Gamma(M)$ and $X \in \NComN$.
	
	Also $\Pi_2 = \cproc{\beta\, X \,\alpha} \parallel \Pi^0_2 \in \sem{\Pi_1}$ and
	$\beta\, X\, \alpha \in \sem{M\,X\,\gamma}$. Thus $\beta \in \sem{M}$ and
	hence $\beta \toSF^* w$, $w \in (\ComSigma \union \ComN)^*$ such that $\M(w) = M$. Then $w \eqvI w_0w_1$ such that $w_0 \in \ComSigma^*$ and $w_1 \in \ComN$ and $M' = \M(w_1)$.

	Hence $w_1\, X\, \alpha \in \sem{M'\,X\,\gamma}$ and thus $\cproc{w_1X \,\alpha} \parallel \Pi^0_2 \oplus \Pi(M) \ChanPar \Gamma \oplus \Gamma(M) := \Pi'_2 \in \sem{\Pi'_1}$.
	
	Using rules \ref{def:conc_faith_send}, \ref{def:conc_faith_spawn} and \ref{def:conc_faith_inter} repeatedly we can see that 
	$\cproc{\beta\, X \,\alpha} \parallel \Pi^0_2  \ChanPar \Gamma \toCF^* 
	 \cproc{w_1X \,\alpha} \parallel \Pi^0_2 \parallel \Pi(w_0) \ChanPar \Gamma \oplus \Gamma(w_0) = 
	 \cproc{w_1X \,\alpha} \parallel \Pi^0_2 \parallel \Pi(w) \ChanPar \Gamma \oplus \Gamma(w) =
	 \cproc{w_1X \,\alpha} \parallel \Pi^0_2 \parallel \Pi(M) \ChanPar \Gamma \oplus \Gamma(M) = \Pi'_2
	 $.
\end{itemize}
Hence the claim holds in all cases and thus we can conclude that
$\simulatedC$ is a weak simulation relation.
\end{proof}

\begin{corollary}
Given an ACPS $P$ if $\cproc{S} \ChanPar \Gamma(\epsilon) \toCM^* \Pi \ChanPar \Gamma$ then
$\cproc{S} \ChanPar \Gamma(\epsilon) \toCF^* \Pi' \ChanPar \Gamma$ such that $\Pi \ChanPar \Gamma \simulatedS \Pi' \ChanPar \Gamma$.
\end{corollary}
\begin{proof}
Follows trivially by induction from Proposition \ref{apx:prop:concurrent_simulation}.
\end{proof}

\setcounter{custthm}{0}
\begin{customtheorem}[Reduction of Program-Point Coverability]
\label{apx:thm:reduction_coverability}
  $(P;l_1, \ldots, l_n)$ is a yes-instance of Program-Point Coverabililty problem iff $(P;l_1, \ldots, l_n)$ is a yes-instance of Alternative Program-Point Coverability problem.
\end{customtheorem}

\begin{proof}
We will first prove the $\implies$-direction.
Let $(P;l_1, \ldots, l_n)$ be a yes-instance of the Program-Point Coverabililty problem then a configuration $\cproc{l_1\alpha_1} \parallel \cdots \parallel \cproc{l_n\alpha_n} \parallel \Pi \ChanPar \Gamma$ for some $\alpha_1, \ldots, \alpha_n \in \CommWords$ is $\toCF$ reachable. By Proposition \ref{apx:prop:conc_reduction_simulation}
$\seqM(\cproc{l_1\alpha_1} \parallel \cdots \parallel \cproc{l_n\alpha_n}) \parallel \seqM(\Pi) \ChanPar \Gamma$ is reachable for $\toCM$ and thus $(P;l_1, \ldots, l_n)$ is a yes-instance of the Alternative Program-Point Coverabililty problem.

For the $\implied$-direction let 
$(P;l_1, \ldots, l_n)$ be a yes-instance of the Alternative Program-Point Coverabililty problem.
Then a configuration $\cproc{\gamma_1} \parallel \cdots \parallel \cproc{\gamma_n} \parallel \Pi \ChanPar \Gamma$ is $\toCM$ reachable and for $i = 1, \ldots, n$  either $\gamma_i = l_i\gamma'_i$ or $\gamma_i = M_i\gamma'_i$ such that $l_i \in M_i$. By Proposition \ref{apx:prop:concurrent_simulation} we can conclude that 
$\cproc{\alpha_1} \parallel \cdots \parallel \cproc{\alpha_n} \parallel \Pi' \ChanPar \Gamma$ is $\toCF$ reachable 
such that 
$\cproc{\gamma_1} \parallel \cdots \parallel \cproc{\gamma_n} \parallel \Pi \ChanPar \Gamma \simulatedC \cproc{\alpha_1} \parallel \cdots \parallel \cproc{\alpha_n} \parallel \Pi' \ChanPar \Gamma$.
That means for $i = 1, \ldots, n$ either $\alpha_i = l_i\alpha'_i$ or $\alpha_i = \beta_i\alpha'_i$
such that $\beta_i \in \ComN^*$ and $\beta_i \toSF^* w_i^0w_i^1$ such that $\M(w_i^1) = M_i$ and $w_i^0 \in \ComSigma^*$. It follows, by using $\eqvI$ where necessary and choosing rewrite rules to expose $l_i$, that $\beta_i \toSF^* {w'}_i^0l_i\beta'_i$  where ${w'}_i^0 \in \ComSigma^*$ and $\beta_i \in \ComN^*$. Hence
$\cproc{\alpha_1} \parallel \cdots \parallel \cproc{\alpha_n} \parallel \Pi' \ChanPar \Gamma \toCF^*
\cproc{\alpha''_1} \parallel \cdots \parallel \cproc{\alpha''_n} \parallel \Pi' \ChanPar \Gamma  \oplus \Gamma({w'}_0^0\cdots {w'}_n^0)$
where either $\alpha''_i = \alpha_i = l_i\alpha_i$ and ${w'}_i^0 = \epsilon$ or
$\alpha''_i = l_i\beta'_i$. Thus we can conclude that $(P;l_1, \ldots, l_n)$ is a yes-instance for the Program-Point Coverability problem.
\end{proof}
\section{Proof of Lemma \ref{lemma:monotonicity_of_inst_chan}}

\begin{lemma}[Sequential Monotonicity]\label{apx:lemma:seq_monotonicity}
The transition relation $\toSM$ is monotone with respect to $\leq_{\Control^{\leq k}}$.
\end{lemma}
\begin{proof}
	Suppose $\gamma, \gamma', \delta \in \Control$ such that $\gamma \leq \delta$ and $\gamma \toSM \gamma'$.
	We will show that there $\exists \delta'$  such that $\delta \toSM \delta'$ and $\gamma \leq \delta'$. 
	We conclude from the definition of $\leq_{\Control^{\leq k}}$ on $\Control^{\leq k}$ and the fact that $\gamma \toSM \gamma'$ that $\gamma = X_1M_1X_2M_2\cdots X_jM_j$ and $\delta = X_1M'_1X_2M'_2\cdots X_jM'_j$ with $M_i \leq_{\Cache} M'_i$ for $1 \leq i \leq j \leq k$.
	Our proof will be by case analysis on $\gamma \toSM \gamma'$.
	\begin{itemize}
	    \item $\gamma \toSM \gamma'$ using Rule \ref{def:seq_mult_sem_nonblock}.\newline
		Thus there is a $X_1 \rightarrow BC$ rule and $C \toSF^* w$ and $\gamma' = B(\M(w) \oplus M_1)X_2M_2\cdots X_jM_j$.
		Hence $\delta \toSM B(\M(w) \oplus M'_1)X_2M'_2\cdots X_jM'_j =: \delta'$. Clearly $(\M(w) \oplus M_1) \leq_{\Cache} (\M(w) \oplus M'_1)$ and thus $\gamma' \leq_{\Control^{\leq k}} \delta'$.
		\item $\gamma \toSM \gamma'$ using Rule \ref{def:seq_mult_sem_block}.\newline
		Thus there is a $X_1 \rightarrow BC$ rule, $C \in \NComN$ and $\gamma' = B\,C\,M_1X_2M_2\cdots X_jM_j$. Further since $\gamma'$ in $\Control^{\leq k}$ it is the case that $j < k$.
		Hence $\delta \toSM B\,C\,M'_1X_2M'_2\cdots X_jM'_j =: \delta'$, $\delta \in \Control^{\leq k}$ since $j < k$ and obviously $\gamma' \leq_{\Control^{\leq k}} \delta'$.
		\item $\gamma \toSM \gamma'$ using Rule \ref{def:seq_mult_sem_tailrec}.\newline
		Thus there is a $X_1 \rightarrow aB$ rule, $a \in \Sigma \union \{\epsilon\}$ and $\gamma' = a\,B\,M_1X_2M_2\cdots X_jM_j$.
		Hence $\delta \toSM a\,B\,M'_1X_2M'_2\cdots X_jM'_j =: \delta'$ and obviously $\gamma' \leq_{\Control^{\leq k}} \delta'$.
		\item $\gamma \toSM \gamma'$ using Rule \ref{def:seq_mult_sem_action}.\newline
		Thus there is a $X_1 \rightarrow a$ rule, $a \in \Sigma \union \{\epsilon\}$ and $\gamma' = a\,M_1X_2M_2\cdots X_jM_j$.
		Hence $\delta \toSM a\,M'_1X_2M'_2\cdots X_jM'_j =: \delta'$ and obviously $\gamma' \leq_{\Control^{\leq k}} \delta'$.
	\end{itemize}
\end{proof}

\begin{lemma}[Monotonicity]\label{apx:lemma:monotonicity_of_inst_chan}
The transition relation $\toCM$ is monotone with respect to $\leq_{\Config}$.
\end{lemma}
\begin{proof}
	Suppose $\Pi_1 \ChanPar \Gamma_1, \Pi'_1 \ChanPar \Gamma'_1, \Pi_2 \ChanPar \Gamma_2 \in \Config$ such that $\Pi_1 \ChanPar \Gamma_1 \leq_{\Config} \Pi_2 \ChanPar \Gamma_2$ and $\Pi_1 \ChanPar \Gamma_1 \toCM \Pi'_1 \ChanPar \Gamma'_1$.
	We will show that there $\exists \Pi'_2$  such that $\Pi_2 \ChanPar \Gamma_2 \toCM \Pi'_2 \ChanPar \Gamma'_2$ and $\Pi'_1 \ChanPar \Gamma'_1 \leq_{\Config} \Pi'_2 \ChanPar \Gamma'_2$. 
	Since $\Pi_1 \ChanPar \Gamma_1 \leq_{\Config} \Pi_2 \ChanPar \Gamma_2$ and 
	$\Pi_1 \ChanPar \Gamma \toCM \Pi'_1 \ChanPar \Gamma'$, we can infer the components of the configuration involved in the latter transition.
	That means $\Pi_1 = \gamma \parallel \Pi^0_1$  and $\Pi_2 = \delta \parallel \Pi^0_2$ such that $\Pi^0_1 \leq_{\M[\Control^{\leq k}]} \Pi^0_2$, $\Gamma_1 \leq_{\Queues} \Gamma_2$, $\gamma, \delta \in \Control^{\leq k}$ and $\gamma \leq_{\Control^{\leq k}} \delta$.
	Our proof will be by case analysis on $\Pi_1 \ChanPar \Gamma \toCM \Pi'_1 \ChanPar \Gamma'_1$.
	\begin{itemize}
		\item $\Pi_1 \ChanPar \Gamma_1 \toCM \Pi'_1 \ChanPar \Gamma'_1$ 
			  using \ref{def:con_mult_sem_interleave}.\newline
			This follows immediately by Lemma \ref{apx:lemma:seq_monotonicity}.
		\item $\Pi_1 \ChanPar \Gamma_1 \toCM \Pi'_1 \ChanPar \Gamma'_1$ 
			  using \ref{def:con_mult_sem_rec}.\newline
			Thus we can conclude 
			\begin{inparaenum}[(i)]
				\item $\gamma = \cproc{\rec{c}{m}\gamma'}$,
				\item $\Pi'_1 = \cproc{\gamma'} \parallel \Pi_1^0$,
				\item $\Gamma_1 = \Gamma'_1 \oplus \Gamma(\snd{c}{m})$. 
				Further since $\Pi_1 \ChanPar \Gamma_1 \leq_{\Config} \Pi_2 \ChanPar \Gamma_2$ we infer
				\item $\delta = \cproc{\rec{c}{m}\delta'}$ with $\gamma' \leq_{\Control^{\leq k}} \delta'$ and
				\item $\Gamma_2 = \Gamma'_2 \oplus \Gamma(\snd{c}{m})$.
			\end{inparaenum}
			
			Then
			 we have $\Pi_2 \ChanPar \Gamma_2 \toCM {\cproc{\delta'} \parallel \Pi^0_2 \ChanPar \Gamma'_2}$.

			Writing $\Pi'_2 := \cproc{\delta'} \parallel \Pi^0_2$ it remains to show $\Pi'_1 \ChanPar \Gamma'_1 \leq_{\Config} \Pi'_2 \ChanPar \Gamma'_2$.

			\begin{inparaenum}[(a)]
				Now 
				\item $\cproc{\gamma'} \leq_{\M[\Control^{\leq k}]} \cproc{\delta'}$,
				\item $\Pi^0_1 \leq_{\M[\Control^{\leq k}]} \Pi^0_2$ by assumption and
				\item since $\Gamma_1 \leq_{\Queue} \Gamma_2$ and clearly $\Gamma'_1 \leq_{\Queue} \Gamma'_2$.
			\end{inparaenum}

			Hence we conclude $\Pi'_1 \ChanPar \Gamma'_1 \leq_{\Config} \Pi'_2 \ChanPar \Gamma'_2$.

		\item $\Pi_1 \ChanPar \Gamma_1 \toCM \Pi'_1 \ChanPar \Gamma'_1$ 
			  using \ref{def:con_mult_sem_spawn}.\newline
			Thus we can conclude 
			\begin{inparaenum}[(i)]
				\item $\gamma = \cproc{(\nu X)\,\gamma'}$,
				\item $\Pi'_1 = \cproc{\gamma'} \parallel \Pi_1^0 \parallel \cproc{X}$,
				\item $\Gamma'_1 = \Gamma_1$. 
				Further since $\Pi_1 \ChanPar \Gamma_1 \leq_{\Config} \Pi_2 \ChanPar \Gamma_2$ we infer
				\item $\delta = \cproc{(\nu X)\,\delta'}$ with $\gamma' \leq_{\Control^{\leq k}} \delta'$.
			\end{inparaenum}
			
			Then
			 we have $\Pi_2 \ChanPar \Gamma_2 \toCM {\cproc{\delta'} \parallel \Pi^0_2 \parallel \cproc{X} \ChanPar \Gamma_2}$.

			Writing $\Pi'_2 := \cproc{\delta'} \parallel \Pi^0_2 \parallel \cproc{X}$ it remains to show $\Pi'_1 \ChanPar \Gamma_1 \leq_{\Config} \Pi'_2 \ChanPar \Gamma_2$.

			\begin{inparaenum}[(a)]
				Now 
				\item $\cproc{\gamma'} \leq_{\M[\Control^{\leq k}]} \cproc{\delta'}$,
				\item $\Pi^0_1 \leq_{\M[\Control^{\leq k}]} \Pi^0_2$ by assumption and
				\item clearly $\cproc{X} \leq_{\M[\Control^{\leq k}]} \cproc{X}$.
			\end{inparaenum}

			Hence we conclude $\Pi'_1 \ChanPar \Gamma'_1 \leq_{\Config} \Pi'_2 \ChanPar \Gamma'_2$.

		\item $\Pi_1 \ChanPar \Gamma_1 \toCM \Pi'_1 \ChanPar \Gamma'_1$ 
			  using \ref{def:con_mult_sem_send}.\newline
			Thus we can conclude 
			\begin{inparaenum}[(i)]
				\item $\gamma = \cproc{\snd{c}{m}\,\gamma'}$,
				\item $\Pi'_1 = \cproc{\gamma'} \parallel \Pi_1^0$,
				\item $\Gamma'_1 = \Gamma_1 \oplus \Gamma(\snd{c}{m})$. 
				Further since $\Pi_1 \ChanPar \Gamma_1 \leq_{\Config} \Pi_2 \ChanPar \Gamma_2$ we infer
				\item $\delta = \cproc{\snd{c}{m}\,\delta'}$ with $\gamma' \leq_{\Control^{\leq k}} \delta'$.
			\end{inparaenum}
			
			Then
			 we have $\Pi_2 \ChanPar \Gamma_2 \toCM {\cproc{\delta'} \parallel \Pi^0_2 \ChanPar 
			 \Gamma_2 \oplus \Gamma(\snd{c}{m})}$.

			Writing $\Pi'_2 := \cproc{\delta'} \parallel \Pi^0_2$ and $\Gamma'_2 := \Gamma_2 \oplus \Gamma(\snd{c}{m})$ it remains to show $\Pi'_1 \ChanPar \Gamma'_1 \leq_{\Config} \Pi'_2 \ChanPar \Gamma'_2$.

			\begin{inparaenum}[(a)]
				Now 
				\item $\cproc{\gamma'} \leq_{\M[\Control^{\leq k}]} \cproc{\delta'}$,
				\item $\Pi^0_1 \leq_{\M[\Control^{\leq k}]} \Pi^0_2$ by assumption and
				\item since $\Gamma_1 \leq_{\Queues} \Gamma_2$, $\oplus$ and $\Gamma(\cdot)$ monotonic we have $\Gamma_1 \oplus \Gamma(\snd{c}{m}) \leq_{\Queues} \Gamma_2 \oplus \Gamma(\snd{c}{m})$.
			\end{inparaenum}

			Hence we conclude $\Pi'_1 \ChanPar \Gamma'_1 \leq_{\Config} \Pi'_2 \ChanPar \Gamma'_2$.

		\item $\Pi_1 \ChanPar \Gamma_1 \toCM \Pi'_1 \ChanPar \Gamma'_1$ 
			  using \ref{def:con_mult_sem_label}.\newline
			Thus we can conclude 
			\begin{inparaenum}[(i)]
				\item $\gamma = \cproc{l\,\gamma'}$,
				\item $\Pi'_1 = \cproc{\gamma'} \parallel \Pi_1^0$,
				\item $\Gamma'_1 = \Gamma_1$. 
				Further since $\Pi_1 \ChanPar \Gamma_1 \leq_{\Config} \Pi_2 \ChanPar \Gamma_2$ we infer
				\item $\delta = \cproc{l\,\delta'}$ with $\gamma' \leq_{\Control^{\leq k}} \delta'$.
			\end{inparaenum}
			
			Then
			 we have $\Pi_2 \ChanPar \Gamma_2 \toCM {\cproc{\delta'} \parallel \Pi^0_2 \ChanPar 
			 \Gamma_2}$.

			Writing $\Pi'_2 := \cproc{\delta'} \parallel \Pi^0_2$ it is trivial to see $\Pi'_1 \ChanPar \Gamma_1 \leq_{\Config} \Pi'_2 \ChanPar \Gamma_2$.

		\item $\Pi_1 \ChanPar \Gamma_1 \toCM \Pi'_1 \ChanPar \Gamma'_1$ 
			  using \ref{def:con_mult_sem_disp}.\newline
			Thus we can conclude 
			\begin{inparaenum}[(i)]
				\item $\gamma = \cproc{M_1X\gamma_1}$,
				\item $\Pi'_1 = \cproc{\gamma'} \parallel \Pi_1^0 \parallel \Pi(M_1)$,
				\item $\gamma' = \cproc{X\gamma_1}$ and
				\item $\Gamma'_1 = \Gamma_1 \oplus \Gamma(M_1)$. 
				Further since $\Pi_1 \ChanPar \Gamma_1 \leq_{\Config} \Pi_2 \ChanPar \Gamma_2$ we infer
				\item $\delta = \cproc{M_2X\delta_1}$ with $M_1 \leq_{\Cache} M_2$ and $X\gamma_1 \leq_{\Control^{\leq k}} X\delta_1$.
			\end{inparaenum}
			
			Let $\delta' := X\delta_1$ then clearly $\gamma' \leq_{\Control^{\leq k}} \delta'$ and
			 we have $\Pi_2 \ChanPar \Gamma_2 \toCM {\cproc{\delta'} \parallel \Pi^0_2 \parallel \Pi(M_2) \ChanPar \Gamma_2 \oplus \Gamma(M_2)}$.

			Writing $\Pi'_2 := \delta' \parallel \Pi^0_2 \parallel \Pi(M_2)$ and $\Gamma'_2 := \Gamma_2 \oplus \Gamma(M_2)$ we will now show $\Pi'_1 \ChanPar \Gamma'_1 \leq_{\Config} \Pi'_2 \ChanPar \Gamma'_2$.

			\begin{inparaenum}[(a)]
				Now 
				\item $\cproc{\gamma'} \leq_{\M[\Control^{\leq k}]} \cproc{\delta'}$,
				\item $\Pi^0_1 \leq_{\M[\Control^{\leq k}]} \Pi^0_2$ by assumption,
				\item since $M_1 \leq_{\Cache} M_2$ and since $\Pi(\cdot)$ is clearly monotonic
						$\Pi(M_1) \leq_{\M[\Control^{\leq k}]} \Pi(M_2)$ and
				\item lastly since $\Gamma(\cdot)$ is monotonic we can conclude
					  $\Pi(M_1) \leq_{\Queues} \Pi(M_2)$.
			\end{inparaenum}

			Hence we conclude $\Pi'_1 \ChanPar \Gamma'_1 \leq_{\Config} \Pi'_2 \ChanPar \Gamma'_2$.
		\item $\Pi_1 \ChanPar \Gamma_1 \toCM \Pi'_1 \ChanPar \Gamma'_1$ 
			  using \ref{def:con_mult_sem_non-term}.\newline
			Thus we can conclude 
			\begin{inparaenum}[(i)]
				\item $\gamma = \cproc{M_1X\gamma_1}$ with $M_1 \in \MixedCache$,
				\item $\Pi'_1 = \cproc{\gamma'} \parallel \Pi_1^0 \parallel \Pi(M_1)$,
				\item $\gamma' = \cproc{M'_1X\gamma_1}$ with $M'_1 \in \NonTermCache$ and $M'_1 = {M_1 \restriction (\NonT \union \PPL)}$,
				\item $\Gamma'_1 = \Gamma_1 \oplus \Gamma(M_1)$. 
				Further since $\Pi_1 \ChanPar \Gamma_1 \leq_{\Config} \Pi_2 \ChanPar \Gamma_2$ we infer
				\item $\delta = \cproc{M_2X\delta_1}$ with $M_1 \leq_{\Cache} M_2$ and $X\gamma_1 \leq_{\Control^{\leq k}} X\delta_1$.
			\end{inparaenum}
			
			Let $M'_2 := M_2 \restriction (\NonT \union \PPL)$ and $\delta' := M'_2X\delta_1$ then
			 we have $\Pi_2 \ChanPar \Gamma_2 \toCM {\cproc{\delta'} \parallel \Pi^0_2 \parallel \Pi(M_2) \ChanPar \Gamma_2 \oplus \Gamma(M_2)}$.

			Writing $\Pi'_2 := \delta' \parallel \Pi^0_2 \parallel \Pi(M_2)$ and $\Gamma'_2 := \Gamma_2 \oplus \Gamma(M_2)$ we will now show $\Pi'_1 \ChanPar \Gamma'_1 \leq_{\Config} \Pi'_2 \ChanPar \Gamma'_2$.

			\begin{inparaenum}[(a)]
				Now 
				since $\cdot \restriction \cdot$ is monotonic in the first argument and $M_1 \leq_{\Cache} M_2$ we conclude $M'_1 \leq_{\Cache} M'_2$ and thus
				\item $\cproc{\gamma'} \leq_{\M[\Control^{\leq k}]} \cproc{\delta'}$,
				\item $\Pi^0_1 \leq_{\M[\Control^{\leq k}]} \Pi^0_2$ by assumption,
				\item and since $\Pi(\cdot)$ is monotonic
						$\Pi(M_1) \leq_{\M[\Control^{\leq k}]} \Pi(M_2)$ and
				\item lastly since $\Gamma(\cdot)$ is monotonic we can conclude
					  $\Pi(M_1) \leq_{\Queues} \Pi(M_2)$.
			\end{inparaenum}

			Hence we conclude $\Pi'_1 \ChanPar \Gamma'_1 \leq_{\Config} \Pi'_2 \ChanPar \Gamma'_2$.
  	\end{itemize}
\end{proof}

\begin{corollary}\label{apx:cor:CM_WSTS}
The transition system ${(\M[\Control^{\leq k}] \times (\Chan \rightarrow \M[\Msg]),\toCM,\leq)}$ 
is a well-structured transition system.
\end{corollary}
\begin{proof}
Follows immediately from Lemma \ref{apx:lemma:monotonicity_of_inst_chan}.
\end{proof}

\begin{theorem}
  The Program-Point Coverability problem for unbounded spawning $k$-ACPS is decidable.
\end{theorem}
\begin{proof}
By Theorem \ref{thm:reduction_coverability} it suffices to show that the Alternative Program-Point Coverability problem is decidable, which follows from Corollary \ref{apx:cor:CM_WSTS}
and the fact that the set 

$$U := \uparrow\{\cproc{\widehat{l_1}X_1^1\cdots X^1_{j_1}} \parallel \cdots \parallel \cproc{\widehat{l_n}X_n^1\cdots X_n^{j_n}} \ChanPar \Gamma(\mset{}) \mid X_i^j \in \NComN, \widehat{l_i} = l_i \text{ or } \mset{l_i} \text{ and } 0 \leq j_i \leq k\}$$

is upward-closed and $(P;l_1,...,l_n)$ is a yes-instance for the Alternative Program-Point Coverability problem iff an element of $U$ is $\toCM$-reachable.
\end{proof}
\newpage
\section{Notation}

\begin{tabular}{l l}
	$M$             	& multiset							\\
	$\mu,\nu$		& general sequence					\\
	$m$ 				& message							\\
	$c$ 				& channel							\\
	$A,B,C,X,Y,Z$ 	\qquad \qquad & non-terminal i.e.~element of $\NonT$						\\
	$l$ 				& label								\\
	$a$ 				& terminal i.e.~element of $\Sigma$		\\
	$w$				& word over $\Sigma$				\\
	$\alpha,\beta$	& word over $(\Sigma \union \NonT)$	\\
	$\gamma,\delta$	& word in $\Control^\M$				\\
	$i,j,n,k$		& integer							\\
	$\Gamma$		& channel 							\\
	$\Pi$			& set of processes					\\
	$U,V,W$			& set 								\\
	$I$				& Independence relation				\\
	$D$				& dependence relation				\\
	$R$				& relation 							\\
	$u,v$			& general element					\\
 	$\toSF$ & standard sequential semantics \\
 	$\toCF$ & standard concurrent semantics \\
 	$\toSM$ & alternative sequential semantics \\
 	$\toCM$ & alternative concurrent semantics \\
	
\end{tabular}

\fi%

\end{document}